\documentclass[12pt]{article}
\usepackage{amsfonts,amssymb,latexsym,amsmath}
\usepackage{eucal}
\usepackage{geometry}
\usepackage{amsmath}
\usepackage{hyperref}
\usepackage{amsthm}
\usepackage{cite}

\usepackage{xcolor}
\usepackage{lipsum}
\usepackage{upgreek}
\usepackage[draft]{todonotes}
\newtheorem{thm}{Theorem}
\newtheorem*{thm*}{Theorem}
\newtheorem{theorem}{Theorem}

\newtheorem{prop}{Proposition}
\newtheorem{coro}{Corollary}
\newtheorem{lemma}{Lemma}
\newtheorem{defn}{Definition}
\newtheorem{remark}{Remark}

\def\<{\langle}
\def\>{\rangle}

\newcommand\be{\begin{equation}} 
\newcommand\ee{\end{equation}}
\newcommand{\comment}[1]

\def\bea{\begin{eqnarray*} }

\def\eea{\end{eqnarray*} }

\begin{document}

\title{Einstein-type elliptic systems}
\author{R. Avalos\thanks{Partially supported by PNPD/CAPES and FUNCAP.}  \,\,and J. H. Lira\thanks
{Partially supported by CNPq and FUNCAP.}}
\date{}
\maketitle

%






\setcounter{MaxMatrixCols}{10}



\begin{abstract}
In this paper we analyse semi-linear systems of partial differential equations which are motivated by the conformal formulation of the Einstein constraint equations coupled with realistic physical fields on asymptotically Euclidean (AE) manifolds. In particular, electromagnetic fields give rise to this kind of system. In this context, under suitable conditions, we prove a general existence theorem for such systems, and, in particular, under smallness assumptions on the free parameters of the problem, we prove existence of far from CMC (near CMC) Yamabe positive (Yamabe non-positive) solutions for charged dust coupled to the Einstein equations, satisfying a trapped surface condition on the boundary. As a bypass, we prove a Helmholtz decomposition on AE manifolds with boundary, which extends and clarifies previously known results.
\end{abstract}

\section{Introduction}

In this article we analyse the Einstein constraint equations (ECE) contemplating the possibility of having many different fundamental fields interacting with gravity. The ECE arise within the initial value formulation of general relativity (GR), where the goal is to construct $(n+1)$-dimensional globally hyperbolic Lorentzian manifolds $(V\doteq M^n\times\mathbb{R},\bar g)$ satisfying the Einstein equations:
\begin{equation}\label{einsteineqs}
\mathrm{Ric}(\bar g)-\frac{1}{2}R(\bar g)\bar g=T(\bar g,\bar \psi)
\end{equation}
by evolving appropriate initial data defined on $M^n$. Above, $\mathrm{Ric}(\bar{g})$ and $R(\bar{g})$ represent the Ricci tensor and Ricci scalar of $\bar g$, respectively, while $T$ represents some $(0,2)$-tensor field, called the energy-momentum tensor, depending on $\bar g$ and (possibly) on a set of fields, collectively denoted by $\bar \psi$, representing other \textit{physical} fields involved in some specific model. 

The geometric picture associated with the above initial value problem consists in the prescription of an initial data set of the form $\mathcal{I}\doteq (M^n,g,K,\epsilon,J)$, where $(M^n,g)$ is an $n$-dimensional Riemannian manifold; $K\in \Gamma(T^0_2M)$ is a symmetric tensor field, while $\epsilon$ and $J$ are, respectively, a function and a 1-form. The problem is then translated into the existence of an isometric embedding $(M^n,g)\hookrightarrow (M^n\times \mathbb{R},\bar{g})$ into a space-time satisfying (\ref{einsteineqs}) such that $K$ corresponds to the extrinsic curvature of $M$ as an embedded hypersurface in space-time, while $\epsilon\doteq T(n,n)\vert_{t=0}$ and $J\doteq -T(n,\cdot)\vert_{t=0}$ correspond to the induced energy-momentum densities on $M$ by the sources $T$, where $n$ stands for the future pointing unit normal vector field to $M$. Straightforwardly, these quantities are seen to be constrained due to the Gauss-Codazzi equations for hypersurfaces, giving rise to the ECE on $M^n$:
\begin{align}\label{Constraints-GC}
\begin{split}
R_g-\vert K\vert^2_{g}+(\mathrm{tr}_gK)^2&=2\epsilon, \\
\mathrm{div}_{g}K-d(\mathrm{tr}_gK)&=J.
\end{split}
\end{align}

In many situations the above constraints stand not only as necessary but also as sufficient conditions for the initial data to admit a well-posed evolution problem \cite{CB2,Ringstrom}. In this context, the system (\ref{Constraints-GC}) has attracted plenty of attention both due to its relevance within GR and its connections with geometric partial differential equation (PDE) problems within geometric analysis. The most common approach to the ECE is the so called conformal method \cite{Lichnerowicz,York1,York2}. This method is used translate the constraints (\ref{Constraints-GC}) into a coupled system consisting of a scalar and a vector equation by splitting the data $(g,K)$ as follows
\begin{align}\label{conf-splitting}
\begin{split}
g=\phi^{\frac{4}{n-2}}\gamma,\;\; K=\phi^{-2}\big(\pounds_{\gamma,conf}X + U \big) + \frac{\tau}{n}g,
\end{split}
\end{align}
where $\gamma$ is some prescribed Riemannian metric, $X\in \Gamma(TM)$,  $\pounds_{\gamma,conf}X\doteq \pounds_X\gamma - \frac{2}{n}\mathrm{div}_{\gamma}X$ is the conformal Killing Laplacian operator and $U$ is a symmetric TT tensor ($\mathrm{tr}_{\gamma}U=0$ and $\mathrm{div}_{\gamma}U=0$). Then, plugging (\ref{conf-splitting}) into the Gauss-Codazzi constraints appearing in (\ref{Constraints-GC}), they take the following form.
\begin{align}\label{constraints.2}
\begin{split}
&\Delta_{\gamma}\phi - c_n R_{\gamma}\phi + c_n \vert\tilde{K}\vert^2_{\gamma}\phi^{-\frac{3n-2}{n-2}} + c_n\left(\frac{1-n}{n}\tau^2 + 2\epsilon  \right)\phi^{\frac{n+2}{n-2}}=0,\\
&\Delta_{\gamma,conf}X -  \left( \frac{n-1}{n}D\tau + J\right) \phi^{\frac{2n}{n-2}} = 0,
\end{split}
\end{align}
where $\Delta_{\gamma}\phi=\mathrm{tr}_{\gamma}(D^2\phi)$ stands for the negative Laplacian; $D$ for the $\gamma$-covariant derivative;  $\Delta_{\gamma,conf}X\doteq\mathrm{div}_{\gamma}(\pounds_{\gamma,conf}X)$; $\tilde{K}\doteq \pounds_{\gamma,conf}X + U$ and $c_n=\frac{1}{4}\frac{n-2}{n-1}$.

When $\epsilon,J,U$ and $\tau$ are treated as known quantities the above system can be thought of as a system on $\phi$ and $X$. Furthermore, if under a conformal transformation $J$ goes to $\phi^{-\frac{2n}{n-2}}\tilde{J}$, where $\tilde{J}$ is treated as a known quantity defined on $(M,\gamma)$ and if $\tau$ is constant, then the momentum constraint decouples from the associated Lichnerowicz equation. In these cases, we can first solve the momentum constraint, and then treat $X$ as a known quantity in the Lichnerowicz equation, which is where the core of the problem relies. Along these lines, on closed manifolds and in the absence of physical sources, J. Isenberg obtained a nice parametrization of the solutions of this problem \cite{Isenberg}. Similar classifications under a constant mean curvature (CMC) condition have been obtained on closed manifolds with physical sources \cite{Maxwell3,CBIP}, compact manifolds with boundary \cite{Holst3} and the relevant Lichnerowicz equation has been analysed on non-compact manifolds, for instance, in \cite{Maxwell1,ChruscielMazzeoCyl,Rigoli1,Rigoli2}. Also, near CMC results have been obtained via implicit function techniques for instance in \cite{Isenberg2,CBIY}, and, in particular in \cite{CBIY}, a constructive approach to produce near CMC solutions for York-scaled sources under curvature assumptions is also available. Summaries of these problems can be reviewed in \cite[Chapter 6]{CB2} and \cite[Chapter 2]{AvalosLiraBook}.



In contrast to the above results, not much was known prior to the work of M. Holst \textit{et al.} in \cite{Holst} about far from CMC solutions for the coupled systems. In this last paper, the authors use a fixed point argument tailored to the Gauss-Codazzi constraint equations with York-scaled sources to guarantee existence of solutions on closed manifolds. The strategy consists mainly in two big steps. On the one hand, linear elliptic estimates are used to produce a priori estimates on solutions to the system, and, in turn, these estimates are used to construct \textit{global barriers} for the relevant Lichnerowicz equation. On the other hand, once these barriers are constructed, the fixed point argument has to be put  to work explicitly. In particular, this procedure allowed the authors to produce far from CMC solutions for Yamabe positive initial data sets, under smallness assumptions on the matter fields and also requiring non-vanishing matter fields. This work was followed by D. Maxwell's in \cite{Maxwell2} where, for instance, these Yamabe postive far from CMC solutions are extended to the vacuum case. After these contributions, similar techniques were used to analyse the conformally formulated Gauss-Codazzi constraints on AE manifolds in \cite{Mazzeo,Holst2}, and related existence results outside the CMC conditions have been obtained by \cite{Dahl,Gicquaud1,Gicquaud2,Nguyen,Premoselli1,Valcu}. Although we will be focused on existence results, problems related with non-uniqueness of solutions in this context have raised plenty of attention, for instance through the following influential papers \cite{Maxwell-nonuniq.1,Maxwell-nonuniq.2,Maxwell-nonuniq.3,Isenberg-nonuniq,Walsh-nonuniq,PR-nonuniq}.


In this context, one further complication arises when the model physical sources appearing in (\ref{einsteineqs}) impose further constraint equations. The most natural example is provided by an electromagnetic field, modelled in space-time by a 2-form $F\in \Omega^2(V)$, which is described by the energy-momentum tensor field
\begin{align*}
T^{EM}_{\alpha\beta}=F_{\alpha}{}^{\mu}F_{\mu\beta} - \frac{1}{4}\bar{g}_{\alpha\beta}F^{\mu\nu}F_{\mu\nu}.
\end{align*}
In particular, the 2-form $F$ must obey the Maxwell equations of electromagnetism on $(V,\bar{g})$, which are given by
\begin{align}\label{MaxwellEqs}
\delta_{\bar{g}}F=\mathcal{J}^{\flat},\:\: dF=0,
\end{align}
where $(\delta_{\bar{g}}F)_{\mu}\doteq-\bar{g}^{\sigma\nu}\bar{\nabla}_{\sigma}F_{\nu\mu}$, $\mathcal{J}$ denotes the charge density associated to the sources of the electromagnetic field and $\mathcal{J}^{\flat}$ its metrically equivalent 1-form. The above equations must be coupled to (\ref{einsteineqs}) and, if one aims to model regions which include the charged matter responsible for sourcing $F$, we must add further energy-momentum contributions. In order to fix a model, let us consider the case of charged dust, where the full energy-momentum tensor is given by
\begin{align}\label{Energy-Momentum}
T=\mu u^{\flat}\otimes u^{\flat} + T^{EM},
\end{align}
where $\mu$ is a scalar function standing for the proper energy density of the fluid, while $u$ is a time-like vector field which stands for the fluid's velocity field. Then, the full system of space-time equations is given by (\ref{einsteineqs})-(\ref{MaxwellEqs}), with the energy-momentum source (\ref{Energy-Momentum}), and, in this setting, the charge density must be given some form. We shall consider the case of a fluid with zero conductivity for simplicity, in which case 
\begin{align*}
\mathcal{J}=qu,
\end{align*}
where $q$ stands for the \emph{proper charge density} of the fluid. In this context, the evolution problem associated to such a system is hyperbolic and results related to it can be consulted, for instance, in \cite[Theorem 3.1, Corollary 3.2 and Theorem 15.1, Chapter IX]{CB2}. In this case, the initial data is subject to a larger constraint system, which now involves two further equations, since each equation in (\ref{MaxwellEqs}) imposes one further constraint.\footnote{See \cite[Chapter 1]{AvalosLiraBook} for a detailed derivation of these constraint equations.} Explicitly, the constraint system now reads
\begin{align}\label{EMConstraints.0}
\begin{split}
R_g-\vert K\vert^2_g+\left( \mathrm{tr}_gK\right)^2 =2\epsilon,\;&\;\; \mathrm{div}_gE=\rho \\
\mathrm{div}_gK -d(\mathrm{tr}_gK)=J,\;&\;\; d\tilde{F}=0\\
\end{split}
\end{align}
where $\tilde{F}$ denotes the restriction of $F$ to tangent vectors to $M$ and $E\doteq F(\cdot,n)$ denotes the electric part of the electromagnetic field, while $\rho=-\mathcal{J}^{\flat}(n)\vert_{t=0}$ denotes the charged density as seen by space-time observers flowing along the integral curves of $n$. Furthermore, straightforward computations give us
\begin{align}
\begin{split}
\epsilon=\mu(Nu^0)^2 + \frac{1}{2}\left( \vert E\vert^2_g + \frac{1}{2}\vert \tilde{F}\vert^2_g\right),\;\; J_k= Nu^{0}\mu u_k - \tilde{F}_{ik}E^i,
\end{split}
\end{align}
where above we have introduced the \emph{lapse function} $N\doteq -\langle \partial_t,n\rangle_{\bar{g}}$, which depends on our choice of time-like observers $t\mapsto (p,t)\in M\times\mathbb{R}\cong V$ and $u^0=-\langle u,e_0\rangle_{\bar{g}}$ denotes the projection of $u$ onto the normal vector to $M$ given by $e_0\doteq Nn$.

Notice that the \emph{magnetic constraint} $d\tilde{F}=0$ is always decoupled from the rest of the system, and demands us to chose this initial datum as a closed 2-form. Thus, from now on, we shall disregard this last equation and assume such a choice has been made. In order to apply the conformal method to the above system, one imposes a conformal splitting for the electric field, given by
\begin{align}\label{electricscalling.1}
E^i=\phi^{-\frac{2n}{n-2}}\tilde{E}^i.
\end{align}
Above, we regard $\tilde{E}$ as the variable we have to solve for in the electric constraint in (\ref{EMConstraints.0}). After analysing how the different sources present in (\ref{EMConstraints.0}) scale under conformal transformations, one finds that the correspoding conformally formulated PDE problem is given by (see \cite[Chapter 2]{AvalosLiraBook} or \cite[Chapter 6]{CB2})
\begin{align}\label{Conformal-EMSystem.1}
\begin{split}
&\!\!\!\!\!\!\!\!\!\!\!\!\!c_n^{-1}\Delta_{\gamma}\phi =R_{\gamma}\phi - \vert\tilde{K}(X)\vert^2_{\gamma}\phi^{-\frac{3n-2}{n-2}} + a_{\tau}\phi^{\frac{n+2}{n-2}} - \vert \tilde{E}\vert^2_{\gamma}\phi^{-3} - \frac{\vert\tilde{F}\vert^2_{\gamma}}{2}\phi^{\frac{n-6}{n-2}},\\
&\!\!\!\!\!\!\!\!\!\!\!\!\!\Delta_{\gamma,conf}X =- \tilde{E}\lrcorner\tilde{F} +  \frac{n-1}{n}d\tau\phi^{\frac{2n}{n-2}}  + \mu\left( 1 + \vert\tilde{u}\vert^2_{\gamma} \right)^{\frac{1}{2}}\tilde{u}^{\flat} \phi^{2\frac{n+1}{n-2}},\\
&\!\!\!\!\!\!\!\!\!\!\!\!\!\mathrm{div}_{\gamma}\tilde{E}=\tilde{q}\phi^{\frac{2n}{n-2}},\\
\end{split}
\end{align}
where above $\tilde{E}\lrcorner\tilde{F}\doteq \tilde{F}(\tilde{E},\cdot)$, $a_{\tau}\doteq \big(2\epsilon_1  - \frac{n-1}{n}\tau^2\big)$ and moving forward we adopt the following notations for the energy-momentum contributions:
\begin{align*}
\epsilon_1=\mu\left( 1 + \vert \tilde{u}\vert^2_{\gamma} \right) \; , \; &\epsilon_2=\frac{1}{2}\vert\tilde{E}\vert^2_{\gamma} \; , \; \epsilon_3=\frac{1}{4}\vert\tilde{F}\vert^2_{\gamma},\\
{\omega_1}_k=\mu\left( 1 + \vert\tilde{u}\vert^2_{\gamma} \right)^{\frac{1}{2}}\tilde{u}_k \; , \; &{\omega_2}_{k}=\tilde{F}_{ik}\tilde{E}^i \; ,\; \tilde{q}=q(1+\vert\tilde{u}\vert^2_{\gamma})^{\frac{1}{2}}.
\end{align*}
Let us highlight that the system (\ref{Conformal-EMSystem.1}) is by nature more subtle than (\ref{constraints.2}), since it can rarely be decoupled without losing its defining properties, \emph{even under a CMC assumption}. Therefore, as long as we want to analyse a realistic charged fluid, we must face the fully coupled system (\ref{Conformal-EMSystem.1}). A similar situation arises when coupling more general Yang-Mills field to the Einstein equations (see, for instance, \cite{CB3,Holm}), but for the sake of clarity, we shall concentrate on the electromagnetic case described above.

Another source of coupling between the conformally formulated constraint equations arises through the imposition of natural boundary conditions. In particular, boundary conditions tailored so as to reproduce black hole initial data sets are of special interest and we shall concentrate on them following ideas of \cite{Holst3,Maxwell1,Dain}. Suppose that $\partial M$ consists of a finite number of compact connected components $\{\Sigma_i\}_{i=1}^{\mathcal{M}}$. Black hole boundary conditions along one such $\Sigma_i=\Sigma$ are typically imposed by the requirement that such a boundary component represents a trapped (or marginally trapped) surface. These conditions are defined by imposing a sing on the \emph{expansion coefficients} $\theta_{\pm}$, which can be written in terms of the initial data set as (see, for instance, \cite{CB2,AvalosLiraBook})
\begin{align}\label{expansioncoefs}
\theta_{\pm}=K(\nu,\nu)-\mathrm{tr}_gK \pm \mathrm{tr}_hk,
\end{align}
where, above, $\nu$ denotes the outward pointing unit normal to $\Sigma$; $h$ denotes the Riemannian metric induced by $g$ on $\Sigma$ and we have defined the extrinsic curvature of $\Sigma$ as a hypersurface of $M$ as $k(X,Y)=g(\nabla_XY,\nu)$, for any vector fields $X,Y$ tangent to $\Sigma$, which implies that $\mathrm{tr}_hk=-\mathrm{div}_g\nu$. The sign conventions we are adopting here are the same as in \cite{AvalosLiraBook,Maxwell1}.

Since the sign of the expansion coefficients measures the divergence of the future light rays  passing through $\Sigma$, among other things, this is why we call $\Sigma$ a \emph{trapped surface} if both $\theta_{\pm}<0$ and \emph{marginally trapped} if $\theta_{\pm}\leq 0$. In fact, the presence of trapped surfaces in non-compact Cauchy surfaces signals the existence of singularities in space-time \cite{Penrose}, which is expected to signal existence black holes in space-time. This expectation is related to the cosmic censorship conjectures, which have a long and rich history in mathematical relativity. If we \emph{freely specify} $\theta_{-}\leq 0$, then these conditions are translated into
\begin{align}\label{BoundaryCond1}
\begin{split}
\mathrm{tr}_hk + \theta_{-} + \tau - K(\nu,\nu) =0,\;\; K(\nu,\nu)\leq \frac{1}{2}\theta_{-} + \tau.
\end{split}
\end{align}
These can be translated into boundary conditions for the conformally formulated equations (\ref{Conformal-EMSystem.1}) applying the conformal splitting described above. In order to satisfy (\ref{BoundaryCond1}), one can impose the following boundary conditions on $(\phi,X,\tilde{E})$:
\begin{align}\label{boundcondsystems}
\begin{split}
&\hat{\nu}(\phi)= -  a_n H \phi  + (d_n\tau + a_n\theta_{-})\phi^{\frac{n}{n-2}} + a_n\left( \frac{1}{2}\vert\theta_{-}\vert - r_{n}\tau\right)v^{\frac{2n}{n-2}}\phi^{-\frac{n}{n-2}},\\
&\pounds_{\gamma,conf}X(\hat{\nu},\cdot)=-\left(\left( \frac{1}{2}\vert\theta_{-}\vert - r_{n}\tau\right)v^{\frac{2n}{n-2}} + U(\hat{\nu},\hat{\nu}) \right)\hat{\nu}^{\flat}, \\
&\langle \tilde{E},\hat{\nu}\rangle_{\gamma}=E_{\hat{\nu}},\\
&\big(\frac{1}{2}\vert\theta_{-}\vert - r_{n}\tau\big)\vert_{\Sigma}\geq 0,\\
&v\geq \phi\vert_{\Sigma},
\end{split}
\end{align}
where the dimensional constant are defined by $a_{n}=\frac{1}{2}\frac{n-2}{n-1}, d_n=\frac{1}{2}\frac{n-2}{n}$ and $r_n=\frac{n-1}{n}$; $H\doteq \mathrm{div}_{\gamma}\hat{\nu}$ is the mean curvature of $\Sigma$ as an embedded hypersurface of $(M,\gamma)$, taken with respect to $-\hat{\nu}$, where $\hat{\nu}$ stands for the outward pointing $\gamma$-unit normal to $\Sigma$, and we have introduced a new function $v$ along the boundary. 

Let us highlight that the last two conditions in (\ref{boundcondsystems}) are introduced so that the boundary condition on $X$ implies the second boundary condition in (\ref{BoundaryCond1}) associated to $\theta_{+}\leq 0$. Also, notice that the constraint $v\geq\phi$ along $\partial M$ is actually a constraint on the solution, if such solution is to represent a (marginally) trapped surface. In practice, there will appear constraints on the admissible data for $v$ such that, via a priori estimates, we can guarantee that this inequality is satisfied by the solutions. Finally, the boundary condition imposed on $\tilde{E}$, which prescribes the conformal data associated to the normal component of the electric field is a physically natural condition from electrodynamics. For further details on these kinds of boundary conditions, besides the previously cited papers, we refer the reader to \cite[Chapter 3]{AvalosLiraBook}.

Having in mind all of the above, our aim is to provide a general and robust framework to analyse the ECE in situations where they must be coupled with matter fields producing a larger system of constraint equations. Physically relevant problems motivating this analysis are given by space-time equations which have a well-posed (short-time existence) initial value formulation for initial data satisfying larger systems of constraints, such as Einstein-Maxwell, Einstein-Yang-Mills and charged Yang-Mills fluids (see \cite{CB2} and \cite{CB3} for a review of several of these results). All of these cases generate additional constraints to (\ref{Constraints-GC}) which have a similar structure to that of (\ref{EMConstraints.0}) and therefore, in general situations, the associated conformally formulated system cannot be decoupled, even under CMC assumptions. We will first present a general existence criteria, based on the existence of special types of barriers, which applies to a class of second order elliptic boundary value problems which substantially generalise (\ref{Conformal-EMSystem.1})-(\ref{boundcondsystems}) preserving its main characteristics, so that it can serve as a basis for existence results for similar systems (see Theorem \ref{ExistenceMetathmInto}). We shall refer to those systems as \emph{Einstein-type elliptic}. Then, we we will state concrete existence theorems for (\ref{Conformal-EMSystem.1})-(\ref{boundcondsystems}) which are based on the explicit construction of the appropriate barrier functions. Besides providing far from CMC initial data for (\ref{Conformal-EMSystem.1})-(\ref{boundcondsystems}) in Theorem \ref{FarCMCthmInto} and near CMC initial data in Theorem \ref{NearCMCthmInto}, it will be seen through the barrier constructions how Theorem \ref{ExistenceMetathmInto} can be put at work and how flexible it actually is. For instance, accommodating a charged perfect fluid becomes an immediate corollary.


\section{Main Results}

In order to summarise our main results, let us first introduce some notation and refer the reader to the main sections for more refined statements. Motivated by the kind of geometric problems commented in the introduction, denote by $E\xrightarrow[]{\pi} M$ a vector bundle over an $n$-dimensional manifold, with $n\geq 3$, and denote by $\psi$ sections of $E$ which belong to some appropriate functional space, say $\mathcal{B}_1$ (in practice, we will use (weighted) Sobolev spaces) . Let $\mathcal{P}:\mathcal{B}_1\mapsto \mathcal{B}_2$ be a linear second order differential operator, where $\mathcal{B}_2$ stands for some other functional space of sections of $E$; and let $\textbf{F}:\mathcal{B}_1\mapsto \mathcal{B}_2$ be a sufficiently regular first order map. We will think about geometric PDEs that can be abstractly written as
\begin{align}\label{abstractpde}
\mathcal{P}(\psi)=\textbf{F}(\psi).
\end{align}  
In this context, we will consider $E=E_0\oplus^N_{j=1} E_j$, with $E_0=M\times \mathbb{R}$ and $E_j\doteq T^{r_j}_{k_j}M$. Therefore, we will always have a second order scalar equation coupled to some system of equations for other fields. The scalar equation will be a Lichnerowicz-type equation. We will consider the case that $\partial M\neq \emptyset$ consists of finitely many compact connected components. In such a case, both $\mathcal{P}$ and $\textbf{F}$ also act on the boundary, producing appropriate boundary conditions. Motivated by PDEs of the form of (\ref{Conformal-EMSystem.1})-(\ref{boundcondsystems}), we will consider that the operators $\mathcal{P}$ belong to some family of operators $\mathcal{P}_{a,b}$ such that $\mathcal{P}=\mathcal{P}_{0,0}$, where $\mathcal{P}_{a,b}$ are invertible for appropriate choices of $a$ and $b$, including $a,b=0$. Here, ``$a$" and ``$b$" stand for appropriate functions introduced so as to \emph{perturb} the Lichnerowicz equation in (\ref{abstractpde}) in a convenient way. With this in mind, we will analyse (\ref{abstractpde}) by first deforming the equation to an equation of the form
\begin{align*}
\mathcal{P}_{a,b}(\psi)=\textbf{F}_{a,b}(\psi),
\end{align*}
where $\textbf{F}_{a,b}$ stands for an appropriate operator constructed from $\textbf{F}$, which shares its mapping properties. Then, denote by $\mathcal{F}_{a,b}$ the map
\begin{align*}
\mathcal{B}_1&\mapsto \mathcal{B}_1,\\
\psi &\mapsto \mathcal{F}_{a,b}(\psi) \doteq \mathcal{P}^{-1}_{a,b}\circ\textbf{F}_{a,b}(\psi).
\end{align*}

Motivated by the work done in \cite{Holst,Maxwell2,Mazzeo,Holst2} we will introduce the idea of \textit{strong global barriers} for systems of the form of (\ref{abstractpde}) (see Definition \ref{strongbarriers}). These are barriers for the associated Lichnerowicz equation which are \textit{uniform} on certain balls $B_{M_{Y^i}}\subset \mathcal{B}_1(M;E_i)$ for all $i=1,\dots,N$. In this context, under appropriate hypotheses on $\mathcal{P}$ and $\textbf{F}$ (see Definition \ref{einsteintypesystems}), in Theorem \ref{conformalexistence} we will prove the following result. 

\begin{theorem}\label{ExistenceMetathmInto}
Suppose that the system (\ref{abstractpde}) is a conformal Einstein-type system of the form of (\ref{generalizedsystem}) posed on a $W^p_{2,\delta}$-AE manifold $(M^n,\gamma)$, with $n\geq 3$, $p>n$ and $-\frac{n}{p}<\delta<n-2-\frac{n}{p}$. Assume that the Lichnerowicz equation admits a compatible pair of strong global sub and supersolutions given by $\phi_{-}$ and $\phi_{+}$, which are, respectively, asymptotic to harmonic functions $\omega_{\pm}$ tending to positive constants $\{A^{\pm}_{j} \}_{j=1}^N$ on each end $\{E_j \}_{j=1}^{N}$. Fix a harmonic function $\omega$ asymptotic to constants $\{A_j \}_{j=1}^{N}$ on each end satisfying $0<A^{-}_j\leq A_j\leq A^{+}_j$, and suppose that the map 
\begin{align*}
\mathcal{F}:\mathcal{B}_1(M;E)&\mapsto \mathcal{B}_1(M;E),\\
\psi=(\varphi,Y)&\mapsto \mathcal{F}(\psi)\doteq \mathcal{P}^{-1}\circ\textbf{F}(\psi).
\end{align*}
is actually invariant on the set $\times_{i}B_{M_{Y^i}}\subset \oplus_{i}\mathcal{B}_1(M,E_i)$ for all $\phi$ such that $\phi_{-}\leq \phi\leq \phi_{+}$. Then, the system admits a solution $(\phi=\omega + \varphi,Y)$, with $(\varphi,Y)\in \mathcal{B}_1(M;E)$, and, furthermore, $\phi>0$. 
\end{theorem}

Regarding the invariance property associated to $\mathcal{F}$ as stated above, we refer the reader to the discussion presented after Theorem \ref{conformalexistence}.

In order to apply this existence result to the concrete case of (\ref{Conformal-EMSystem.1})-(\ref{boundcondsystems}) without loosing generality, we will appeal to a Helmholtz decomposition for the electric field. Such a decomposition on AE manifolds without boundary was presented in \cite{Cantor}. In Theorem \ref{Helmoholtzdecomposition}
we will establish an appropriate analogue for such non-compact manifolds with compact boundary. Explicitly, we will show the following.

\begin{theorem}\label{HelmholtzThmInto}
Let $(M^n,g)$ be an $n$-dimensional $W^p_{2,\delta}$-AE, with $n\geq 3$, $p>n$ and $\delta>-\frac{n}{p}$. If $-\frac{n}{p}<\rho<n-1-\frac{n}{p}$, $\rho-1\neq -\frac{n}{p}$, and one of the following two conditions holds\\
1) $-\frac{n}{p}<\rho-1$, or\\
2) $M$ has only one end,\\
then the following decomposition holds:
\begin{align*}
W^p_{1,\rho}(M;TM)\times W^p_{1-\frac{1}{p}}(\partial M)=\nabla^{N}_g(W^p_{2,\rho-1})\oplus\mathrm{Ker}(\mathcal{L}_{2}).
\end{align*}
\end{theorem}

In the above theorem $\nabla^N:W^p_{2,\rho-1}(M)\mapsto W^p_{1,\rho}(M)\times W^p_{1-\frac{1}{p}}(\partial M)$ is given by $\nabla^Nu=(\nabla_gu,\nu(u))$, where $\nu$ stands for the outward point unit normal to $\partial M$, and $\mathcal{L}_2:W^p_{1,\rho}(TM)\times W^p_{1-\frac{1}{p}}(\partial M)\mapsto L^p_{\rho+1}(M)\times W^p_{1-\frac{1}{p}}(\partial M)$ is given by $\mathcal{L}_2=(\mathrm{div}_g,\mathrm{Id})$. Above, conditions (1) and (2) are sharp in the following sense: If $M$ has two ends and condition (1) fails, then the decomposition cannot hold. We would like to stress that the above theorem is of interest on its own right.

After applying the above decomposition in the electromagnetic case, we will show in Lemma \ref{EMcompactness} that the system (\ref{Conformal-EMSystem.1})-(\ref{boundcondsystems}) belongs to the class of systems where Theorem \ref{ExistenceMetathmInto} applies. In this context, after producing appropriate strong global barriers for such a system, in Theorem \ref{Yamabepostiveexistence}, we will prove a far from CMC existence result, which basically establishes the following.

\begin{theorem}\label{FarCMCthmInto}
Let $(M^n,\gamma)$ be a $W^p_{2,\delta}$-Yamabe positive AE manifold with compact boundary $\partial M$, with $n\geq 3$, $p>n$ and $\delta>-\frac{n}{p}$. Then, for freely specified mean curvature $\tau\in W^{p}_{1,\delta+1}$, under smallness assumptions on the remaining coefficients of the system (\ref{Conformal-EMSystem.1})-(\ref{boundcondsystems}), there is a $W^p_{2,\delta}$-solution to the conformal problem (\ref{Conformal-EMSystem.1})-(\ref{boundcondsystems}). 
\end{theorem}


In the same spirit, in Theorem \ref{Yamabenonpostiveexistence} we will prove the following near-CMC result for metrics which are not Yamabe positve.

\begin{theorem}\label{NearCMCthmInto}
Let $(M,\gamma)$ be a $W^p_{2,\delta}$-AE manifold with compact boundary $\partial M$, with $n\geq 3$, $p>n$ and $\delta>-\frac{n}{p}$. Under smallness hypotheses on the coefficients of (\ref{Conformal-EMSystem.1})-(\ref{boundcondsystems}) which include a near-CMC condition, there is a $W^p_{2,\delta}$-solution to the conformal problem (\ref{Conformal-EMSystem.1})-(\ref{boundcondsystems}).
\end{theorem}

\section{Analysis of the PDE system}
Since during most of this paper we will be mainly concerned AE manifolds, we will now begin with a couple of definitions and analytic results concerning AE manifolds where we will be following the notations established in \cite{CB2,Cantor,CB-C,McOwen}. Let us merely notice that some of the main results we will derive apply also to compact manifolds, modulo technical hypotheses, since, for many purposes, we can think about compact manifolds as AE manifolds with no ends. 
\begin{defn}[Manifolds euclidean at infinity]
A complete $n$-dimensional smooth Riemannian manifold $(M,e)$ is called euclidean at infinity if there is a compact set $K$ such that $M\backslash K$ is the disjoint union of a finite number of open sets $U_i$, such that each $U_i$ is isometric to the exterior of an open ball in Euclidean space.   
\end{defn}
Notice that any manifold $M^n$ with the topological structure implied by the above definition carries a smooth Riemannian metric which is isometric to the euclidean metric in each end sufficiently near infinity. We denote such a metric by ``$e$" and use it to define some functional spaces.

\begin{defn}[Weighted Sobolev spaces]
A weighted Sobolev space $W^p_{s,\delta}$, with $s$ a nonnegative integer and $\delta\in\mathbb{R}$, of tensor fields of some given type on the manifold $(M,e)$ Euclidean at infinity, is the space of tensor fields with generalized derivatives of order up to $s$ in the metric $e$ such that $D^mu(1 + d^2)^{\frac{1}{2}(m+\delta)}\in L^p$, for $0\leq m\leq s$. It is a Banach space with norm
\begin{align}\label{norm}
\Vert u\Vert^{p}_{W^p_{s,\delta}}\doteq \sum_{0\leq m\leq s}\int_M\vert D^mu\vert_e^p\sigma^{p(m+\delta)}\mu_{e}
\end{align}
where $D$ represents the $e$-covariant derivative; $\mu_e$ the Riemannian volume form associated with $e$, $\sigma(x)=(1+d^2(x))^{\frac{1}{2}}$ and $d(x)=d(x,p)$ the distance in the Riemannian metric $e$ of an arbitrary point $x$ to a fixed point $p$.
\end{defn}


Noticing that, near infinity, a field $u\in L^p_{\delta}\cap C^0$ behaves like $\vert u\vert_e=o(\vert x\vert^{-(\delta+\frac{n}{p})})$ makes clear the relation between our choice of weight parameter and other choices used in other standard references, such as \cite{Maxwell1,Mazzeo,Holst2,Bartnik}. Furthermore, having in mind the behaviour at infinity of elements in $W^p_{s,\delta}$, we present the following definition.

\begin{defn}[AE manifolds]
Let $(M,e)$ be a manifold euclidean at infinity and $g$ be a Riemannian metric on $M$. Then we will say that $(M,g)$ is $W^p_{s,\delta}$-AE if $g-e\in W^p_{s,\delta}$ for some $\delta>-\frac{n}{p}$ and $s>\frac{n}{p}$.
\end{defn}

Let us now adapt the above language so as to allow that within the compact region $K$ we have a boundary.
Notice that we have a finite number of end charts, say $\{U_i,\varphi_i \}_{i=1}^{k}$, with $\varphi(U_i)\simeq \mathbb{R}^n\backslash \overline{B_1(0)}$, and a finite number of coordinate charts covering the compact region $K$, say $\{U_i,\varphi_i\}_{i=k+1}^N$. We can consider a partition of unity $\{\eta_i\}_{i=1}^N$ subordinate to the coordinate  cover $\{U_i,\varphi_i\}_{i=1}^N$, and let $V_i$ be equal to either $\mathbb{R}^n$ or $\mathbb{R}^n_{+}$, depending on whether $U_i$, $i\geq k+1$, is an interior of boundary chart respectively. Then, given a vector bundle $E\xrightarrow{\pi} M$, we can define $W^p_{s,\delta}(M;E)$ to be the subset of $W^p_{s,loc}(M;E)$ such that 
\begin{align}\label{locsob}
\Vert u\Vert_{W^p_{s,\delta}}&=\sum_{i=1}^k\Vert {\varphi^{-1}_i}^{*}(\eta_i u)\Vert_{W^p_{s,\delta}(\mathbb{R}^n)} + \sum_{i=k+1}^N\Vert {\varphi^{-1}_i}^{*}(\eta_i u)\Vert_{W^p_{s}(V_i)}<\infty. 
\end{align}

Similarly, we introduce the weighted $C^{k}_{\beta}$-spaces, $k\geq 0$ an integer and $\beta\in \mathbb{R}$, to be the subset of $C^{k}_{loc}(M;E)$ such that 
\begin{align}\label{Ck-norm}
\Vert u\Vert_{C^k_{\beta}}=\sum_{j\leq k}\sup_M\vert D^ju\vert_e(1+d^2_e)^{\frac{1}{2}(\beta+j)}<\infty,
\end{align}
which is also a Banach space under this norm. In this context the following properties hold (see \cite{CB2,Maxwell1,Cantor,CB-C,Bartnik}), where we assume we may have finitely many compact boundary connected components $\{\Sigma_i\}_{i=1}^{\mathcal{M}}\subset K$.
\begin{prop}[$W^p_{s,\delta}$-Properties]\label{sobprop}
Let $M^n$ be a manifold euclidean at infinity, $n\geq 3$. Then, for $1<p<\infty$ a real number, $\delta,\delta',\delta_1,\delta_2,\beta\in \mathbb{R}$, and $s,m$ non-negative integers the following properties hold:\\
i) If $1\leq s<\frac{n}{p}$, then $W^{p}_{s+m,\delta}\hookrightarrow W^{\frac{np}{n-sp}}_{m,\delta+s}$,\\
ii) If $s>\frac{n}{p}$, then $W^{p}_{s+m,\delta}\hookrightarrow C^{m}_{\beta}$, with $\beta\leq \delta+\frac{n}{p}$.\\
iii) If $s_1,s_2\geq s$; $s_1+s_2>s+\frac{n}{p}$ and $\delta< \delta_1 + \delta_2 + \frac{n}{p}$, then the multiplication mapping $(u,v) \rightarrow u\otimes v$, defines a continuous bilinear map between $W^{p}_{s_1,\delta_1}\times W^p_{s_2,\delta_2} \rightarrow W^{p}_{s,\delta}$.\\
iv) The embedding $W^p_{s,\delta}\hookrightarrow W^p_{s-1,\delta'}$ is compact for any $s\geq 1$ and $\delta'<\delta$.
\end{prop}

Using the Sobolev embeddings we get that if $p>n$, then $\vert u\vert_e\lesssim \sigma^{-(\delta+\frac{n}{p})}\Vert u\Vert_{W^p_{1,\delta}}$, with $\sigma\doteq (1+d^2_e)^{\frac{1}{2}}$. Notice that if $r$ is a continuous function which in the ends, sufficiently near infinity, agrees with the euclidean radial function $\vert x\vert$, then there are constants $C_1$ and $C_2$, such that $C_1r(x)\leq \sigma(x)\leq C_2 r(x)$. In particular, this implies that one can substitute $r(x)$ for $\sigma$ in (\ref{Ck-norm}) to define an equivalent norm. Furthermore, using the above relations, we get the following.
\begin{prop}\label{Mazzeo3}
Let $(M,g)$ be a $W^p_{2,\delta}$-AE manifold, and consider a function $r$ on $M$ which near infinity, in each end, agrees with $\vert x\vert$. Then, if $p>n$, for any $u\in W^p_{1,\delta}$ the following estimate holds
\begin{align}
\vert u\vert_e\lesssim r^{-(\delta+\frac{n}{p})}\Vert u\Vert_{W^p_{1,\delta}}.
\end{align} 
\end{prop}

The following property will be important for us, and since we have not found it explicitly stated in any standard reference, we will present a proof. 
\begin{prop}\label{compactembed1}
Let $(M,e)$ be a manifold euclidean at infinity and consider the same setting as above. Then, it holds that the embedding $W^p_{1,\delta}\hookrightarrow C^0_{\delta'+\frac{n}{p}}$ is compact for any $\delta'<\delta$ and $p>n$. 
\end{prop}
\begin{proof}
Let $R\gg 1$ be a real number, consider spheres $S^i_R$ of radius $R$ contained in each end $E_i$ of $M$ and denote by $B_R$ the open subset of $M$ covering the compact core and such that $\partial B_{R}=\cup_iS^i_{R}$. Since $M$ is complete, then $\overline{B_{2R}}$ is a compact manifold with smooth boundary. Then, consider a cut off function $\chi_R$, which is equal to one on $\overline{B_R}$, equal to zero on $M\backslash\overline{B_{2R}}$ and $0\leq \chi_R\leq 1$. Then, let $\{u_k\}\subset W^p_{1,\delta}$ be a bounded sequence, \textit{i.e}, $\Vert u_k\Vert_{W^p_{1,\delta}}\leq 1$ for all $k$, and split $u_k=\chi_Ru_k + (1-\chi_R)u_k$. Since $\chi_Ru_k$ is a bounded sequence supported in $\overline{B_{2R}}$, then $\{\chi_Ru_k\}\subset W^p_{1}(\overline{B_{2R}})$. Since for $p>n$ $W^p_{1}(\overline{B_{2R}})$ is compactly embedded in $C^0(\overline{B_{2R}})$, there is a subsequence $\{\chi_Ru_{k_{j}}\}_{j=1}^{\infty}$ which is convergent in $C^0(\overline{B_{2R}})$, to which we now restrict. Now, using a partition of unity as in (\ref{locsob}), for $R$ large enough, we write $(1-\chi_{R})u_{k_j}=\sum_i\eta_i(1-\chi_{R})u_{k_j}$. Since these fields are supported in the ends of $M$ and can be extended as fields on $\mathbb{R}^n$, appealing the discussion prior to Proposition \ref{Mazzeo3} let us below consider the equivalent $C^{k}_{\beta}$-norm taken with $r(x)=(1+\vert x\vert^2)^{\frac{1}{2}}$, so that 
\begin{align*}
\Vert (1-\chi_{R})u_{k_j}-(1-\chi_{R})&u_{k_l}\Vert_{C^0_{\delta'+\frac{n}{p}}}\leq \sum_i\sup_{\mathbb{R}^n}\vert {\varphi^{-1}_i}^{*}\big((1-\chi_{R})(\eta_iu_{k_j}-\eta_iu_{k_l})\big)\vert_er^{\delta'+\frac{n}{p}}\\
&\leq (1+R^2)^{-\frac{1}{2}(\delta-\delta')}\sum_i\sup_{\mathbb{R}^n}\vert {\varphi^{-1}_i}^{*}\eta_iu_{k_j} - {\varphi^{-1}_i}^{*}\eta_iu_{k_l}\vert_er^{\delta +\frac{n}{p}} 
\end{align*} 
Now, under our hypotheses, we have a continuous embedding $W^p_{1,\delta}\hookrightarrow C^0_{\delta+\frac{n}{p}}$, which implies that there is a constant $C>0$, depending only on $M$, such that
\begin{align*}
\Vert (1-\chi_{R})u_{k_j}-(1-\chi_{R})u_{k_l}\Vert_{C^0_{\delta'+\frac{n}{p}}}&\leq C(1+R^2)^{-\frac{1}{2}(\delta-\delta')}\Vert u_n-u_m\Vert_{W^p_{1,\delta}(M)},\\
&\leq 2C(1+R^2)^{-\frac{1}{2}(\delta-\delta')}.
\end{align*}
Now, fix $\epsilon>0$ and, since $\delta-\delta'>0$, there is a radius $R_{\epsilon}$ sufficiently large such that $(1+R_{\epsilon}^2)^{-\frac{1}{2}(\delta-\delta')}< \frac{\epsilon}{4C}$. Once we have fixed such an $R_{\epsilon}>0$, from the above discussion, we know that $\{{\chi_R}_{\epsilon}u_{k_j} \}\subset W^p_1(\overline{B_{2R_{\epsilon}}})$ is Cauchy in $C^0(\overline{B_{2R_{\epsilon}}})$, which implies that there is an $N=N(\epsilon)$, such that $\forall$ $j,l\geq N$ it holds that $\Vert\chi_R u_{k_j} - \chi_R u_{k_l}\Vert_{C^0(\overline{B_{2R_{\epsilon}}})}< \frac{\epsilon}{2}$. Furthermore, since on $\overline{B_{2R_{\epsilon}}}$ the $C^0$ and $C^0_{\delta'+\frac{n}{p}}$ norms are equivalent, this last statement also holds with the $C^0(\overline{B_{2R_{\epsilon}}})$ norm changed by $C^0_{\delta'+\frac{n}{p}}(\overline{B_{2R_{\epsilon}}})$.  Thus, we get that
\begin{align*}
\Vert u_{k_j}-u_{k_l}\Vert_{C^0_{\delta'+\frac{n}{p}}(M)}&<\frac{\epsilon}{2} + \frac{\epsilon}{2}=\epsilon,
\end{align*} 
proving that there is a subsequence which is Cauchy in $C^0_{\delta'+\frac{n}{p}}$, and thus is convergent. 
\end{proof}

\begin{coro}\label{compactembed2}
Let $(M,e)$ be a manifold euclidean at infinity and consider the same setting as above. Then, it holds that the embedding $W^p_{2,\delta}\hookrightarrow C^1_{\delta'+\frac{n}{p}}$ is compact for any $\delta'<\delta$ and $p>n$.
\end{coro}
Let us furthermore comment that, given a compact manifold with compact boundary, $W^{p}_{k}$-fields have $W^p_{k-\frac{1}{p}}$ traces on the boundary, and the trace map is continuous (see \cite{Holst3} for a review on this topic and \cite{Adams,Grisvard} for a more detailed discussions). In our present case, since the trace map only acts on the compact core of the manifold, we get the same result. Furthermore, we should notice that, given a compact manifold $\Sigma$, the continuous embedding property $W^p_{s}(\Sigma)\hookrightarrow C^k(\Sigma)$ still holds for any real $s$ satisfying $s-k>\frac{n}{p}$, and $1<p<\infty$ (see \cite{Holst3,Adams,Grisvard}). 

As we will shortly see, the PDE system associated with a charged fluid given by (\ref{Conformal-EMSystem.1})-(\ref{boundcondsystems}), involves only two different differential operators with nice linear properties. These are the Laplacian acting on functions and the conformal Killing Laplacian acting on vector fields. The main properties we will need concerning these operators are enclosed in the following two propositions, which follow from results established, for instance, in \cite{Maxwell2,Cantor}.

\begin{prop}\label{isomorphismthm}
Let $(M,\gamma)$ be a $W^p_{2,\rho}$-AE manifold with $p>n$ and $\rho>-\frac{n}{p}$. Consider the operators
\begin{align}
\begin{split}
\mathcal{P}_1:W^{p}_{2,\delta}(M,\mathbb{R})&\mapsto L^p_{\delta+2}(M,\mathbb{R})\times W^p_{1-\frac{1}{p}}(\partial M,\mathbb{R}),\\
u&\mapsto (\Delta_{\gamma}u - au, -(\nu(u) + bu )\vert_{\partial M}),\\
\mathcal{P}_2:W^{p}_{2,\delta}(M,TM)&\mapsto L^p_{\delta+2}(M,T^{*}M)\times W^p_{1-\frac{1}{p}}(\partial M,T^{*}M),\\
X&\mapsto (\Delta_{\gamma,conf}X,\pounds_{\gamma,conf}(\nu,\cdot)\vert_{\partial M})
\end{split}
\end{align}
with $a\in L^p_{\delta+2}, b\in W^p_{1-\frac{1}{p}}$ non-negative functions, $\nu$ the outward pointing normal to $\partial M$ and $\delta>-\frac{n}{p}$. Then, if $-\frac{n}{p}<\delta<n-2-\frac{n}{p}$, both these operators are isomorphism.
\end{prop}
Above, the condition $p>n$ instead of $p>\frac{n}{2}$ is only a sufficient condition to guarantee that the conformal Killing Laplacian is an isomorphism. This could be replaced by $p>\frac{n}{2}$ and demanding that $\gamma$ possesses no conformal Killing fields (CKF). Due to \cite[Theorem 4]{Maxwell1}, we know that if $\gamma$ is a $W^p_{2,\rho}$-AE metric with $p>n$ and $\rho>-\frac{n}{p}$, then $\gamma$ does not admits any CKF in $W^p_{2,\delta}$ for any $\delta>-\frac{n}{p}$. From the above theorem plus standard arguments we can establish the following elliptic estimates, which will be crucial in our analysis.
\begin{prop}\label{injectivityestimate}
Consider the same set up as in the above proposition and for any $u\in W^p_{2,\delta}(M;\mathbb{R})$ and $X\in W^p_{2,\delta}(M;TM)$ write $\mathcal{P}_1u=(f_u,g_u)\in L^p_{\delta+2}(M,\mathbb{R})\times W^p_{1-\frac{1}{p}}(\partial M,\mathbb{R})$ and $\mathcal{P}_2X=(Y_X,Z_X)\in L^p_{\delta+2}(M,T^{*}M)\times W^p_{1-\frac{1}{p}}(\partial M,T^{*}M)$. Then, there are constants $C_1,C_2>0$ such that the following estimates hold for any $u\in W^p_{2,\delta}$ and any $X\in W^p_{2,\delta}$:
\begin{align}
\begin{split}
\Vert u\Vert_{W^p_{2,\delta}}&\leq C_1\big(\Vert f_u\Vert_{L^p_{\delta+2}} + \Vert g_u\Vert_{W^p_{1-\frac{1}{p}}} \big),\\
\Vert X\Vert_{W^p_{2,\delta}}&\leq C_2\big(\Vert Y_X\Vert_{L^p_{\delta+2}} + \Vert Z_X\Vert_{W^p_{1-\frac{1}{p}}} \big).
\end{split}
\end{align}
\end{prop}

\subsection{Fixed point system}

We shall now start by rewriting our PDE system (\ref{Conformal-EMSystem.1})-(\ref{boundcondsystems}) as an elliptic system such that we can construct solutions by iteration. In order to rewrite the electric constraint as a second order equation, we will apply a Helmholtz decomposition to the electric field in order to decompose it as the sum of an exact and co-exact 1-forms. In the case of AE-manifolds without boundary such a decomposition was addressed in \cite[Theorem 7.6]{Cantor}. Our aim is to establish a valid analogous decomposition in the case of manifolds with boundary. We will take some time to do this, since this is an interesting result on its own, and, during this process, we will clarify some points concerning the result established in \cite{Cantor}. We will appeal to two well-established functional analytic results. The following one appears in \cite[Lemma 2.2.]{Cantor}

\begin{thm}\label{Banachsplitting}
Let $T:X\mapsto Y$ and $S:Y\mapsto Z$ be bounded linear operators between Banach spaces. Then, the following are equivalent:
\begin{align*}
&1) \: \mathrm{Ker}(S\circ T)=\mathrm{Ker}(T) \text{ and } \mathrm{Im}(S\circ T)=\mathrm{Im}(S).\\
&2) \: Y=\mathrm{Im}(T)\oplus\mathrm{Ker}(S).
\end{align*}
Furthermore, in case one of the above holds, then $\mathrm{Im}(T)$ is closed in $Y$.
\end{thm}

Also, we will need the following classical result obtained by putting together Theorem 4.7, Theorem 4.12 and Theorem 4.14 of \cite{Rudin}. 
\begin{thm}\label{adjoint-surjectivity}
Suppose that $T:X\mapsto Y$ is a bounded linear transformation between reflexive Banach spaces and that $T$ has closed range. Then, if $T^{*}:Y^{*}\mapsto X^{*}$ denotes the adjoint of $T$, it holds that $\mathrm{Ker}(T)^{\perp}=T^{*}(Y^{*})$. Furthermore, the range of $T^{*}$ is closed and so $\mathrm{Ker}(T^{*})^{\perp}=T(X)$.
\end{thm}

In this setting, referring the reader to the notations established in the introduction after Theorem \ref{HelmholtzThmInto}, we shall establish the following result.

\begin{thm}[Helmholtz decomposition]\label{Helmoholtzdecomposition}
Let $(M^n,g)$ be an $n$-dimensional $W^p_{2,\delta}$-AE manifold, with $n\geq 3$, $p>n$ and $\delta>-\frac{n}{p}$. If $-\frac{n}{p}<\rho<n-1-\frac{n}{p}$, $\rho-1\neq -\frac{n}{p}$, and one of the following two conditions hold\\
1) $-\frac{n}{p}<\rho-1$, or\\
2) $M$ has only one end,\\
then the following decomposition holds:
\begin{align}\label{Hdecomposition}
W^p_{1,\rho}(M;TM)\times W^p_{1-\frac{1}{p}}(\partial M)=\nabla^{N}_g(W^p_{2,\rho-1})\oplus\mathrm{Ker}(\mathcal{L}_{2}).
\end{align}
\end{thm}

\begin{proof}

Let us begin by showing that if $-\frac{n}{p}<\rho<n-1-\frac{n}{p}$, $\rho-1\neq \frac{n}{p}$, and either (1) or (2) hold, then $\mathrm{Ker}(\mathcal{L})=\mathrm{Ker}(\mathcal{L}_1)$. We only need to show the inclusion $\mathrm{Ker}(\mathcal{L})\subset \mathrm{Ker}(\mathcal{L}_1)$, since the other one is trivial. Also, notice that the relevant inclusion is not trivial only because $\rho>-\frac{n}{p}$ and therefore $\mathcal{L}$ can have kernel in $W^p_{2,\rho-1}$, whereas if $\rho-1>-\frac{n}{p}$ this last kernel is trivial. Therefore, on any $W^p_{2,\delta}$-AE manifold with any number of ends, if $\rho-1>-\frac{n}{p}$ then $\mathrm{Ker}(\mathcal{L})=\mathrm{Ker}(\mathcal{L}_1)$. Thus, from now on, we will restrict our analysis to the case $-\frac{n}{p}-1<\rho-1<-\frac{n}{p}$ and $M$ having just one end. Notice that in this interval for $\rho$, $\mathcal{L}$ has closed range as a consequence of arguments similar to those of \cite[Theorem 1.10]{Bartnik}.\footnote{In order to achieve this result in the case of manifolds with boundary one would need to complement the elliptic estimate of \cite[Theorem 1.10]{Bartnik} with the corresponding estimates on the compact core, which follows along the same lines to that of \cite[Proposition 4]{Maxwell1}.} Thus, consider $u\in \mathrm{Ker}(\mathcal{L})\subset W^p_{2,\rho-1}$, and a covering of $M$ by coordinate charts $\{U_i,\varphi_i \}_{i=1}^{m+1}$, where $\{U_i,\varphi_i \}_{i=1}^{m}$ denote the covering for the compact core of the manifold and $(U_{m+1},\varphi_{m+1})$ denotes the end chart. Let $\{\eta_{i} \}_{i=1}^{m+1}$ be a partition of unity subordinate to this cover and therefore write $u=\sum_{i=1}^{m}\eta_iu + \eta_{m+1}u$. From now on, we will distinguish the end chart and just denote $\eta_{m+1}=\eta$. Then, notice that
\begin{align*}
0=\eta\Delta_gu=\Delta_g(\eta u) - u\Delta_g\eta - 2g(\nabla \eta,\nabla u).
\end{align*}  
Therefore, defining $\bar{u}=\eta u$, we get that
\begin{align*}
\Delta_g\bar{u} = u\Delta_g\eta + 2g(\nabla \eta,\nabla u)\doteq F.
\end{align*}
Notice that, via identification with the end chart, the above expression is defined on $E\cong \mathbb{R}^n\backslash B_1$, and $F$ is compactly supported in $E$. Therefore, we can extend $F$ to $\mathbb{R}^n$ by declaring it to be zero inside $B_1$. This implies that $F\in W^p_1(\mathbb{R}^n)$ is compactly supported in $E$, which means that $F\in W^p_{1,\tilde{\sigma}}$ for any $\tilde{\sigma}\in\mathbb{R}$. 

In this context, given $\sigma\in\mathbb{R}$, if we consider the action of $\Delta_g:W^p_{2,\sigma}\mapsto L^p_{\sigma+2}$, then $\Delta^{*}_{g}:L^{p'}_{-\sigma-2}\mapsto W^{p'}_{-2,-\sigma}$ is injective if $-(\sigma+2)>-\frac{n}{p'}=-n(1-\frac{1}{p})$, that is $\sigma<n-2-\frac{n}{p}$. Thus, from Theorem \ref{adjoint-surjectivity}, $\Delta_g:W^p_{2,\sigma}\mapsto L^p_{\sigma+ 2}$ is surjective for any $\sigma<n-2-\frac{n}{p}$ and, since $F\in L^p_{\sigma+2}$ for any such $\sigma$, then, there exists $v\in W^p_{2,\sigma}(\mathbb{R}^n)$ for any $\sigma<n-2-\frac{n}{p}$ satisfying $\Delta_gv=F$, which, in turn, implies that $\Delta_g(\bar{u} - v)=0$.

Now, in \cite[Proposition 2.2]{Bartnik} it is established that on any $W^p_{2,\delta}$-AE manifold without boundary and with $p>n$, it holds that if $-\frac{n}{p}-1<\rho-1<-\frac{n}{p}$, then 
\begin{align}
\mathrm{dim}(\mathrm{Ker}(\Delta_g:W^p_{2,\rho-1}\mapsto L^p_{\rho+1}))=N_0(\rho-1),
\end{align}
where $N_0(\rho-1)$ is the dimension of the space of harmonic polynomials of degree up to $k^{-}(\rho-1)$ of the euclidean Laplacian, where $k^{-}(\rho-1)$ equals the maximum \textit{exceptional} integer\footnote{Following \cite{Bartnik}, the exceptional integers are defined as $\{z\in\mathbb{Z}, z\neq -1,\cdots,3-n \}$.} value $k$ satisfying $-k>\rho-1+\frac{n}{p}$, which, in our case, is $k^{-}(\rho-1)=0$. Therefore, $N_0(\rho-1)=1$. That is, $\mathrm{Ker}(\Delta_g:W^p_{2,\rho-1}\mapsto L^p_{\rho+1})$ is a 1-dimensional vector space, which is parametrized by the constants.

All of the above implies that $\bar{u}-v=c$ for some constant and some function and $v\in W^p_{2,\sigma}$ for any $\sigma<n-2-\frac{n}{p}$. Notice that this implies that $u-c\in W^p_{2,\sigma}$ for any such $\sigma$. Therefore, we can write
\begin{align*}
u=c+\tilde{u}, \;\; \tilde{u}\in W^p_{2,\sigma}.
\end{align*}
This implies that $u$ is a harmonic function with zero Neumann boundary data on $\partial M$, which is asymptotic to a constant function $c$ as we go to infinity in $E$. Therefore, \cite[Proposition A.3]{Holst2} implies that such $u$ is uniquely determined by its asymptotic values, and furthermore it satisfies $\min c\leq u\leq \max c$, that is $u=c$ and therefore $u\in \mathrm{Ker}(\mathcal{L}_1)$.

Now, let us show that if $\rho<n-1-\frac{n}{p}$, then $\mathrm{Im}(\mathcal{L})=\mathrm{Im}(\mathcal{L}_2)$. Similarly to the above claim, we only need to show that $\mathrm{Im}(\mathcal{L}_2)\subset\mathrm{Im}(\mathcal{L})$, since the other inclusion is trivial. In order to do this, let us show that, in fact, under our hypotheses, $\mathcal{L}$ is surjective, which is even stronger. Using Theorem \ref{adjoint-surjectivity}, we know that this is equivalent to proving that 
\begin{align*}
\mathcal{L}^{*}:L^{p'}_{-(\rho+1)}(M)\times W^{p'}_{-(1-\frac{1}{p})}(\partial M)\mapsto \left(W^{p}_{2,(\rho-1)}(M)\right)'
\end{align*}
is injective. Therefore, consider $(u,v)\in \mathrm{Ker}(\mathcal{L}^{*})$. Then, for any $\phi\in W^p_{2,\rho-1}$, it holds that
\begin{align}\label{adjoint0}
0&=\langle(u,v),\mathcal{L}\phi) \rangle =\int_Mu\Delta_g\phi\mu_g + \langle v,\nu(\phi)\rangle_{W^{p'}_{-(1-\frac{1}{p})}(\partial M)\times W^{p}_{1-\frac{1}{p}}(\partial M)}
\end{align}
In particular, the above holds for all $\phi\in C^{\infty}_0(\overset{\circ}{M})$, and for those test functions we get
\begin{align*}
0=\langle \Delta_gu,\phi\rangle \;\; \forall \;\; \phi\in C^{\infty}_{0}(\overset{\circ}{M}),
\end{align*}
where $\Delta_gu$ is to be understood in the sense of distributions. This implies that $\Delta_gu=0$ on $\overset{\circ}{M}$, with $u\in L^{p'}_{loc}$. Therefore, from elliptic regularity, we get $u\in W^{p'}_{2,loc}\cap L^{p'}_{-(\rho+1)}$ with $\Delta_gu=0$ and  hence $u\in W^{p'}_{2,-(\rho+1)}$.\footnote{Here we are using \cite[Theorem 3.1]{NW} and the same observation applied in \cite[Theorem 4.1]{CB-C}. That is, the hypotheses of \cite[Theorem 3.1]{NW} can be weakened so as to suppose $u\in W^p_{k,loc}$ instead of $W^p_{k}$.} Then, after integrating by parts (\ref{adjoint0}), we get that the following holds. 
\begin{align}\label{adjoint1}
0&=\langle v,\nu(\phi)\rangle_{W^{p'}_{-(1-\frac{1}{p})}(\partial M)\times W^{p}_{1-\frac{1}{p}}(\partial M)} + \int_{\partial M}\{u\nu(\phi) - \phi\nu(u)\}dS
\end{align}

Now, suppose the following condition holds:\\

\noindent $\star$) For any $\theta\in W^{p}_{2-\frac{1}{p}}(\partial M)$ and any $\chi\in W^{p}_{1-\frac{1}{p}}(\partial M)$ there is an element $\phi\in W^{p}_{2,\rho-1}(M)$ such that $\phi\vert_{\partial M}=\theta$ and $\nu(\phi)\vert_{\partial M}=\chi$.\\

Then, from (\ref{adjoint1}), we would get that $\forall \theta\in W^{p}_{2-\frac{1}{p}}(\partial M) \text{ and } \chi\in W^{p}_{1-\frac{1}{p}}(\partial M)$
\begin{align}
0&=\langle v,\chi \rangle_{W^{p'}_{-(1-\frac{1}{p})}(\partial M)\times W^{p}_{1-\frac{1}{p}}(\partial M)} + \int_{\partial M}\{u\chi - \theta\nu(u)\}dS\;\; 
\end{align}
This, in turn, by taking $\chi=0$ would imply that $\nu(u)\vert_{\partial M}=0$. Therefore, we would get $\mathcal{P}(u)\doteq (\Delta_gu,\nu(u)\vert_{\partial M})=0$ and $u\in W^{p'}_{2,-(\rho+1)}$. But, from \cite{Maxwell1}, we know that $\mathcal{P}$ is injective if $-(\rho+1)>-\frac{n}{p'}=-n+\frac{n}{p}$. That is, $\mathcal{P}$ is injective for $\rho<n-1-\frac{n}{p}$. Therefore, $u=0$. This implies that
\begin{align*}
0=\langle v,\chi \rangle_{W^{p'}_{-(1-\frac{1}{p})}(\partial M)\times W^{p}_{1-\frac{1}{p}}(\partial M)}  \;\; \forall  \chi\in W^{p}_{1-\frac{1}{p}}(\partial M).
\end{align*}
Thus, $v=0$, which implies that $\mathrm{Ker}(\mathcal{L}^{*})=\{ 0\}$. Therefore, we if we prove $(\star)$, the initial claim follows. In order to do this, first notice that from \cite{Adams}, we know that the trace map $\gamma: W^q_{2}(\mathbb{R}^n)\mapsto W^q_{2-\frac{1}{q}}(\mathbb{R}^{n-1})\times W^q_{1-\frac{1}{q}}(\mathbb{R}^{n-1})$ given by $u\mapsto (u,\nu(u))\vert_{x^n=0}$ defines a continuous isomorphism between $W^q_{2}(\mathbb{R}^n)/\mathrm{Ker}(\gamma)\mapsto W^q_{2-\frac{1}{q}}(\mathbb{R}^{n-1})\times W^q_{1-\frac{1}{q}}(\mathbb{R}^{n-1})$ for any $1<q<\infty$. Since the trace map only acts in a neighbourhood of the boundary, via a partition of unity argument and the use of adapted coordinates to the boundary, we can extend this to a continuous isomorphism between $W^{q}_{2,\rho-1}(M)/\mathrm{Ker}(\gamma)\mapsto W^q_{2-\frac{1}{q}}(\partial M)\times W^q_{1-\frac{1}{q}}(\partial M)$. This, in particular, implies that given arbitrary $(\theta,\chi)\in W^{p}_{2-\frac{1}{p}}(\partial M)\times W^{p}_{1-\frac{1}{p}}(\partial M)$, there is a (non-unique) $\phi\in W^{p}_{2,\rho-1}(M)$ such that $\gamma \phi=(\theta,\chi)$. This proves that $(\star)$ holds and therefore our second claim holds. Finally, putting all of the above together with Theorem \ref{Banachsplitting}, we get the Helmholtz decomposition (\ref{Hdecomposition}).
\end{proof}

\begin{remark}
The above decomposition is sharp with respect to condition (2). That is, if $M$ has two ends, and $-\frac{n}{p}<\rho<-\frac{n}{p}+1$, then $\mathrm{Ker}(\mathcal{L}_1)\varsubsetneq\mathrm{Ker}(\mathcal{L})$. This can be seen from the following balancing-type argument. Consider two different constants $c_1$ and $c_2$ and define $\omega\doteq \eta_1c_1 + \eta_2c_2$, where $\eta_i$ are cut off functions supported on each end $E_i$ respectively, which equal $1$ in a neighbourhood of infinity. Then, since $\Delta_g\omega\in L^p_{\sigma+2}$ for any $\sigma<n-2-\frac{n}{p}$, we know from Proposition \ref{isomorphismthm} that there is an $\omega'\in W^p_{2,\sigma}$ such that $\Delta_g\omega'=\Delta_g\omega$ and $\nu(\omega)=0$ along $\partial M$. Then, define $u\doteq \omega-\omega'$, which implies that
\begin{align}\label{HelmoltzSharp}
\begin{split}
\Delta_gu&=0,\\
\nu(u)&=0 \text{ along } \partial M.
\end{split}
\end{align} 
Now, notice that $\omega'\in W^p_{2,\sigma}\hookrightarrow W^p_{2,\rho-1}$ and, since $\rho<-\frac{n}{p}+1$, constant functions belong to $L^p_{\rho-1}$, which implies that $\omega\in W^p_{2,\rho-1}$. Therefore $u\in W^p_{2,\rho-1}$ and satisfies (\ref{HelmoltzSharp}). That is, $u\in \mathrm{Ker}(\mathcal{L})$. But, notice that $u$ cannot be constant, since it is asymptotic to two different constants at infinity in each end. Therefore, $u\not\in \mathrm{Ker}(\mathcal{L}_1)$, and hence Theorem \ref{Banachsplitting} implies that the decomposition cannot hold. 
\end{remark}

The above theorem implies that choosing $\rho=\delta+1$, with $-\frac{n}{p}<\delta<n-2-\frac{n}{p}$, we can decompose $\tilde{E}\in W^p_{1,\delta+1}$ as $\tilde{E}=df + \vartheta$, for some $f\in W^p_{2,\delta}$ and $\vartheta\in W^p_{1,\delta}$ such that $\mathrm{div}_{\gamma}\vartheta=0$. Taking this into account, we see that the system (\ref{Conformal-EMSystem.1})-(\ref{boundcondsystems}) now reads as a semi-linear second order PDE system for $(\phi,f,X)$, where the electric constraint is translated into the boundary value problem
\begin{align}\label{ElectricPotentialEq}
\begin{split}
\Delta_{\gamma}f = \tilde{q}\phi^{\frac{2n}{n-2}},\;\; \hat{\nu}(f)=E_{\hat{\nu}}.
\end{split}
\end{align}
  
From now on, when referring to the system (\ref{Conformal-EMSystem.1})-(\ref{boundcondsystems}), we will consider this second order boundary value problem for the fully coupled system were the above Helmholtz decomposition for the electric constraint is understood. Also, regarding the Lichnerowicz equation, we are looking for bounded solutions with some prescribed asymptotic behaviour. We shall account for this as follows. First, fix some positive constants $\{ A_j\}_{j=1}^N$ which are meant to represent the asymptotic values of $\phi$ on each end $\{E_j\}_{j=1}^N$ respectively. Then, let $\omega$ be the unique $W^{p}_{2,loc}$ function, $p>n$, satisfying
\begin{align}
\begin{split}
\Delta_{\gamma}\omega&=0,
\end{split}
\end{align}
and such that $\omega\rightarrow A_j \text{ as we move to infinify in } E_j$, which is guaranteed to exist due to \cite[Proposition A.3]{Holst2}. Then, we write $\phi=\omega+\varphi$, for some $\varphi\in W^{p}_{2,\delta}$ to be found by solving the system (\ref{Conformal-EMSystem.1})-(\ref{boundcondsystems}). With this in mind, define the vector bundle $E\doteq (M\times\mathbb{R})\oplus(M\times\mathbb{R})\oplus TM$, and consider $W^p_{2,\delta}$-sections of this vector bundle. Then, we get the following differential operator
\begin{align*}
\begin{split}
&\mathcal{P}:W^p_{2,\delta}(M;E)\mapsto L^p_{\delta+2}(M;E)\times W^p_{1-\frac{1}{p}}(\partial M; E),\\
&(\varphi,f,X)\mapsto (\Delta_{\gamma}\varphi,\Delta_{\gamma}f,\Delta_{\gamma,conf}X, -\nu(\varphi)\vert_{\partial M},-\nu(f)\vert_{\partial M},\pounds_{\gamma,conf}X(\nu,\cdot)\vert_{\partial M})
\end{split}
\end{align*}
Now, denote by $\textbf{F}$ the map taking $(\phi,f,X)\rightarrow \textbf{F}(\phi,f,X)$, where $\textbf{F}(\phi,f,X)$ stands for the function appearing in the right hand side of (\ref{Conformal-EMSystem.1})-(\ref{boundcondsystems}). In this setting, we rewrite the system more compactly as
\begin{align}\label{fixedpointsys.2}
\mathcal{P}(\psi)=\textbf{F}(\psi),
\end{align}
where $\psi\in W^p_{2,\delta}(M;E)$. At this point, the idea is to solve the above problem by solving a sequence of linear problems. Such a procedure requires us to first \emph{perturb} the Lichnerowicz equation so as to preserve the invertibility associated to $\mathcal{P}$ and further gain some necessary monotonicity properties that will allow to control the solutions along the iteration process appealing to maximum principles. In the following section we shall show how this is done for the case of the system (\ref{Conformal-EMSystem.1})-(\ref{boundcondsystems}).

\subsection{Shifted system}

Following the discussion presented above at the end of the previous section, let us now consider the the following \emph{shifted} system:
\begin{align}\label{shiftedsystem}
\begin{split}
\Delta_{\gamma}\varphi - a\varphi&= c_n R_{\gamma}\phi - c_n \vert\tilde{K}\vert^2_{\gamma}\phi^{-\frac{3n-2}{n-2}} - c_n\left(2\epsilon_1  - \frac{n-1}{n}\tau^2\right)\phi^{\frac{n+2}{n-2}} - 2c_n\epsilon_2\phi^{-3} \\
&- 2c_n\epsilon_3\phi^{\frac{n-6}{n-2}} - a\varphi ,\\
\Delta_{\gamma}f &= \tilde{q}\phi^{\frac{2n}{n-2}},\\
\Delta_{\gamma,conf}X  &= \frac{n-1}{n}D\tau \phi^{\frac{2n}{n-2}} + \omega_1\phi^{2\frac{n + 1}{n-2}} - \omega_2,
\end{split}
\end{align}
with boundary conditions:
\begin{align}\label{shiftedboundarycond}
\begin{split}
&\hat{\nu}(\varphi) - b\varphi= -  a_n H \phi  + (d_n\tau + d_n\theta_{-})\phi^{\frac{n}{n-2}} + a_n\left( \frac{1}{2}\vert\theta_{-}\vert - r_{n}\tau\right)v^{\frac{2n}{n-2}}\phi^{-\frac{n}{n-2}} - b\varphi,\\
&\hat{\nu}(f)=E_{\hat{\nu}},\\
&\pounds_{\gamma,conf}X(\hat{\nu},\cdot)=-\left(\left( \frac{1}{2}\vert\theta_{-}\vert - r_{n}\tau\right)v^{\frac{2n}{n-2}} + U(\hat{\nu},\hat{\nu}) \right)\hat{\nu}^{\flat},
\end{split}
\end{align}
with $a\in L^p_{\delta+2}(M)$, $b\in W^p_{1-\frac{1}{p}}(\partial M)$ satisfying $a,b\geq 0$, $\phi=\omega+\varphi$ and $\omega$ is a harmonic function with zero Neumann boundary conditions which captures the behaviour of $\phi$ at infinity, as described in the previous section. We will denote the linear operator appearing in the left-hand side by
\begin{align*}
&\mathcal{P}_{a,b}:W^p_{2,\delta}\mapsto L^p_{\delta+2}(M,E)\times W^{p}_{1-\frac{1}{p}}(\partial M,E),\\
&(\varphi,f,X)\mapsto (\Delta_{\gamma}\varphi-a\varphi,\Delta_{\gamma}f,\Delta_{\gamma,conf}X, -(\nu(\varphi)+b\varphi)\vert_{\partial M},-\nu(f)\vert_{\partial M},\pounds_{\gamma,conf}X(\nu,\cdot)\vert_{\partial M}),
\end{align*}
and $\textbf{F}_{a,b}(\psi)$ by the right-hand side of (\ref{shiftedsystem})-(\ref{shiftedboundarycond}). Furthermore, we will constraint the choices of $\theta_{-}$ and $\tau$ so as to satisfy the constraint (\ref{boundcondsystems}), and we need to show that, given some $v\in W^p_{1-\frac{1}{p}}(\partial M)$, the solutions of the above boundary value problem satisfy $(v-\phi)\vert_{\partial M}\geq 0$, so as to satisfy the marginally trapped surface condition. Then, we can rewrite the shifted system as
\begin{align*}
\mathcal{P}_{a,b}(\psi)=\textbf{F}_{a,b}(\psi).
\end{align*}
Notice that the operator $\mathcal{P}_{a,b}$ as defined above is invertible for $-\frac{n}{p}<\delta<n-2-\frac{n}{p}$, so that fixing some $\psi_0\in W^p_{2,\delta}$, the sequence $\{\psi_k \}_{k=0}^{\infty}\subset W^p_{2,\delta}$ given inductively by $\psi_{k+1}\doteq \mathcal{P}^{-1}_{a,b}(\textbf{F}_{a,b}(\psi_k))$ is well-defined. Furthermore, continuity of both $\mathcal{P}_{a,b}$ and $\textbf{F}_{a,b}$ implies that, if we can extract a $W^p_{2,\delta}$-convergent subsequence with limit $\psi$, then this limit will solve $\mathcal{P}(\bar{\psi})=\textbf{F}(\bar{\psi})$. Now, since $(\Delta_{\gamma}\varphi,\hat{\nu}(\varphi)\vert_{\partial M})=(\Delta_{\gamma}\phi,\hat{\nu}(\phi)\vert_{\partial M})$, we see that such procedure provides us with a solution to the full constraint system with marginally trapped boundary conditions.

The above paragraph points towards the procedure we should apply to our shifted system. Considering the sequence $\{\psi_k \}_{k=0}^{\infty}$ defined inductively as above, in order to extract a convergent subsequence, we can appeal to the compact embedding in $C^1_{\delta'+\frac{n}{p}}$, for $\delta'<\delta$. In order to do this, we need to show that this sequence is uniformly bounded in $W^p_{2,\delta}$. If this is guaranteed, then the compact embedding gives us a limit function $\psi$, such that $\psi_k\xrightarrow[]{} \psi$ in $C^1_{\delta'+\frac{n}{p}}$. Finally, the aim is to prove that elliptic estimates imply that the limit actually holds in $W^p_{2,\delta}$. This lays out the steps towards proving existence of solutions to the coupled constraint system. In order to deal with the first of these steps, we introduce the following concept, which is inspired in the ideas of global barriers given in \cite{Holst,Maxwell2}. 

\begin{defn}\label{strongbarriers}
Consider the Lichnerowicz equation associated to the conformal problem for the Einstein constraint equations, and let us write it as follows.
\begin{align}\label{generalizedlich}
\begin{split}
\Delta_{\gamma}\phi&=\sum_{I}a_I(Y)\phi^{I},\\
-\hat{\nu}(\phi)&=\sum_Jb_J(Y)\phi^{J} \text{ on } \partial M
\end{split}
\end{align}
where $\gamma\in W^p_{2,\delta}$; $\hat{\nu}$ is the outward pointing unit normal with respect to $\gamma$; $Y=(Y_1,\cdots,Y_r)$,with $\{Y_i\}_{i=1}^r$ being a set of tensor fields, denotes the set of fields involved as unknowns in the conformal problem besides the conformal factor and ``$I$" and ``$J$" denote the exponents which define the non-linearities of the Lichnerowicz equation. We will say that $\phi_{-}$ is a \textbf{strong global subsolution} if there are positive numbers $\{M_{Y_i}\}_{i=1}^r$ such that
\begin{align}
\begin{split}
\Delta_{\gamma}\phi_{-}&\geq\sum_{I}a_I(Y)\phi_{-}^{I}\;\;  \;\; \forall \;\; Y\in \times_iB_{M_{Y_i}}\\
-\hat{\nu}(\phi_{-})&\geq \sum_Jb_J(Y)\phi_{-}^{J} \text{ on } \partial M
\end{split}
\end{align}
where $B_{M_{Y_i}}\subset W^p_{2,\delta}$ denotes the closed ball in $W^p_{2,\delta}$ of radius $M_{Y_i}$. A \textbf{strong global supersolution} is defined in the same way with the opposite inequality. Also, if the same set of numbers $\{B_{M_{Y_i}}\}_{i=1}^r$ serve for both the sub and supersolution, and $0< \phi_{-}\leq\phi_{+}$ we will say the the barriers are \textbf{compatible}. 
\end{defn} 


Since the system associated to a charged fluid serves as model for several other physically relevant situation which introduce further coupling between the conformally formulated ECE, we will now broaden the scope of our analysis so as to present an existence criteria based on the existence of strong global barriers which applies to more general systems. Careful analysis actually reveals that this is done only at the expense of carrying along some heavier notations, without further technical difficulties. In particular, we can have in mind boundary value problems of the form: 
\begin{align}\label{generalizedsystem}
\begin{split}
\Delta_{\gamma}\phi&=\sum_{I}a^0_I(Y)\phi^{I},\\
L^{i}(Y^i)&=\sum_{J}a^{i}_{J}(Y)\phi^{J}, \;\; i=1,\cdots,r,\\
-\hat{\nu}(\phi)&=\sum_{K}b^0_K(Y)\phi^{K}, \text{ on } \partial M\\
B^i(Y^i)&=\sum_{L} b^{i}_{L}(Y)\phi^{L}, \;\; i=1,\cdots,r \text{ on } \partial M,
\end{split}
\end{align}
where $(L^{i},B^i)$ represent continuous linear elliptic second order operators with boundary conditions, acting between $W^p_{2,\delta}(M)\mapsto L^p_{\delta+2}(M)\times W^p_{1-\frac{1}{p}}(\partial M)$, which are invertible for $p>n$ and $\delta>-\frac{n}{p}$; $a^{0}_{I},a^{i}_{J}:\times_k W^p_{2,\delta}(M,E_k)\mapsto L^p_{\delta+2}(M,E_{\alpha})$ and $b^{0}_{K},b^{i}_{L}:\times_k W^p_{2,\delta}(M,E_k)\mapsto W^p_{1-\frac{1}{p}}(\partial M,E_{\alpha})$ stand for continuous maps between these spaces, which can depend on $\{Y_i\}$ and $\{DY_i\}$, with $\alpha=0,i$. We will also impose that the coefficients $a^0_{I}$ and $b^{0}_K$ satisfy the following \textit{boundedness} property. Let $M_Y=\sum_iM_{Y^i}$, where $B_{M_{Y^i}}\subset W^p_{2,\delta}(M;E_i)$ denotes the closed ball of radius $M_{Y^i}>0$, then there are functions $f_I\in L^p_{\delta+2}(M)$, $g_K\in W^p_{1-\frac{1}{p}}(\partial M)$ and constants $C_K>0$, independent of $Y$, such that
\begin{align}\label{boundednessprop}
\begin{split}
\vert a^0_I(Y)\vert&\leq f_I \text{ for any } Y\in B_{M_Y},\\
\vert b^0_K(Y)\vert&\leq g_K \text{ for any } Y\in B_{M_Y},\\
\Vert b^0_K(Y)\Vert_{W^p_{1-\frac{1}{p}}}&\leq C_K \text{ for any } Y\in B_{M_Y},
\end{split}
\end{align}
and, furthermore, we will impose that the coefficients $a^{0}_I, b^{0}_{K}, a^{i}_{J}$ and $b^{i}_L$ satisfy the following \textit{compactness} property: Given a bounded $W^p_{2,\delta}$-sequence $\{Y_k \}_{k=1}^{\infty}$ and $-\frac{n}{p}<\delta'<\delta$, if $Y_{k}\xrightarrow[]{} Y$ in $C^1_{\delta'+\frac{n}{p}}$, then it holds that 
\begin{align}\label{compactnessprop}
\begin{split}
a^{\alpha}_{I}(Y_k)\xrightarrow[k\rightarrow \infty]{} a^{\alpha}_{I}(Y) \text{ in } L^p_{\delta+2}\; \text{ and }\; b^{\alpha}_{J}(Y_k)\xrightarrow[k\rightarrow \infty]{} b^{\alpha}_{J}(Y) \text{ in } W^p_{1-\frac{1}{p}}.
\end{split}
\end{align}
Notice that this last property is motivated by the compact embedding $W^p_{2,\delta}\hookrightarrow C^1_{\delta'+\frac{n}{p}}$ for any $p>n$ and $\delta>\delta'$. In order to show that these properties are sensible, let us consider the following lemma.
\begin{lemma}\label{EMcompactness}
Let $(M,\gamma)$ be a $W^p_{2,\delta}$-AE manifold with $p>n$ and $\delta>-\frac{n}{p}$ and consider the system (\ref{Conformal-EMSystem.1})-(\ref{boundcondsystems}) with the electric field $\tilde{E}$ given as in (\ref{ElectricPotentialEq}). Suppose that the prescribed data for the problem satisfies the functional hypotheses $\tilde{u}\in W^{p}_{1,\delta}(M)$, $\mu,\tilde{q}\in W^p_{1,\delta+2}(M)$, $U,\tau,\vartheta,\tilde{F}\in W^p_{1,\delta+1}(M)$, $H,\theta_{-},E_{\hat{\nu}}\in W^p_{1-\frac{1}{p}}(\partial M)$ and $v\in W^p_{2-\frac{1}{p}}(\partial M)$. Then, such system is of the type (\ref{generalizedsystem}) satisfying all the properties required for the coefficients.
\end{lemma}
\begin{proof}
In this case $E=(M\times\mathbb{R})\oplus(M\times\mathbb{R})\oplus TM$ and $Y=(f,X)$ and we just have to check the mapping properties of the coefficients. Notice that if $\mu\in W^p_{1,\delta+2}$, then so are $\epsilon_1$ and $\omega_1$. Also, under our hypotheses on $\delta$ and $p$ we get that $2\epsilon_3=\vert\tilde{F}\vert^2_{\gamma}\in W^p_{1,\delta+2}$ from the multiplication property and the same holds for $\tau^2\in W^p_{1,\delta+2}$. Clearly, we also have that $R_{\gamma},d\tau\in L^p_{\delta+2}$, and all these coefficients do not depend on $Y$. Now, let us consider the coefficients $\tilde{K},\epsilon_2$ and $\omega_2$. Notice that
\begin{align*}
\vert\tilde{K}\vert^2_{\gamma}&=\vert\pounds_{\gamma,conf}X\vert^2_{\gamma} + 2\langle U,\pounds_{\gamma,conf}X \rangle_{\gamma} + \vert U\vert^2_{\gamma},\\
2\epsilon_2&=\vert df\vert^2_{\gamma} + 2\langle \vartheta,df \rangle_{\gamma} + \vert\vartheta\vert^2_{\gamma},\\
\omega_2(Y)&=\tilde{F}(Df,\cdot) + \tilde{F}(\vartheta,\cdot)
\end{align*}
The multiplication property implies that all these coefficients are in $W^p_{1,\delta+2}$. Let us begin by checking the boundedness property (\ref{boundednessprop}). Notice that this property is clear for the coefficients $\epsilon_1,\tau^2,\epsilon_3$ and $R_{\gamma}$ which all belong to $L^p_{\delta+2}$ and are independent of $Y$. Also, from Lemma \ref{Mazzeo3}, a function $u\in W^p_{1,\delta+1}$ satisfies $\vert u\vert\lesssim r^{-(\delta+1+\frac{n}{p})}\Vert u\Vert_{W^p_{1,\delta+1}}$. Then, in the case of $\vert\tilde{K}\vert^2_{\gamma}$ and $\epsilon_2$, notice that we have the following.
\begin{align*}
\vert\tilde{K}(Y)\vert^2_{\gamma}&=\vert\pounds_{\gamma,conf}X\vert^2_{\gamma} + 2\langle U,\pounds_{\gamma,conf}X \rangle_{\gamma} + \vert U\vert^2_{\gamma}\lesssim \vert DX\vert^2_{\gamma} + \vert U\vert_{\gamma}\vert DX\vert_{\gamma} + \vert U\vert^2_{\gamma},\\
&\lesssim r^{-2(\delta+1+\frac{n}{p})}(\Vert DX\Vert^2_{W^p_{1,\delta+1}} + \Vert U\Vert_{W^p_{1,\delta+1}}\Vert DX\Vert_{W^p_{1,\delta+1}} + \Vert U\Vert^2_{W^p_{1,\delta+1}}),\\
&\lesssim (M^2_Y + \Vert U\Vert_{W^p_{1,\delta+1}}M_Y + \Vert U\Vert^2_{W^p_{1,\delta+1}})r^{-(\delta+2+\frac{n}{p})}r^{-(\delta+\frac{n}{p})} \;\; \forall \;\; Y\in B_{M_Y}.
\end{align*}
Since under our hypotheses $r^{-(\delta+2+\frac{n}{p})}r^{-(\delta+\frac{n}{p})}\in L^p_{\delta+2}$, we see that $\vert\tilde{K}\vert_{\gamma}^2$ satisfies the boundedness property. A similar result holds for $\epsilon_2$, implying that they both satisfy (\ref{boundednessprop}). Now, let us examine this property for the boundary coefficients. First notice that none of these coefficients depend on $f$ or $X$. Also, from our choices of functional spaces, we know that $U$ and $\tau$ have $W^p_{1-\frac{1}{p}}$-traces on $\partial M$ and $\theta_{-},H\in W^p_{1-\frac{1}{p}}$. All this together implies that $b^{\alpha}_{J}(Y)\in W^p_{1-\frac{1}{p}}$ for all $J$ and $\alpha=0,i$, and these coefficients are actually independent of $Y$. Therefore, the boundedness property holds.

Now, let us check the compactness property. Thus, considering a bounded sequence $\{Y_k\}_{k=1}^{\infty}\subset W^p_{2,\delta}$, such that $Y_k\xrightarrow[]{} Y$ in $C^1_{\delta'+\frac{n}{p}}$, we get that
\begin{align*}
\!\!\!\big\vert \vert\tilde{K}(Y)\vert^2_{\gamma}-\vert\tilde{K}(Y_k)\vert^2_{\gamma}\big\vert
&\leq \vert\langle \pounds_{\gamma,conf}X,\pounds_{\gamma,conf}(X - X_k) \rangle_{\gamma}\vert + \vert\langle \pounds_{\gamma,conf}X_k,\pounds_{\gamma,conf}(X - X_k) \rangle_{\gamma}\vert\\
& + 2\vert\langle U,\pounds_{\gamma,conf}(X - X_k) \rangle_{\gamma}\vert,\\
&\lesssim \vert DX\vert_e\vert D(X-X_k)\vert_e + \vert DX_k\vert_e\vert D(X-X_k)\vert_e + \vert U\vert_e\vert D(X-X_k)\vert_e,\\
&\lesssim r^{-(\delta+1+\frac{n}{p})}r^{-(\delta'+1+\frac{n}{p})}\left( \Vert DX\Vert_{W^p_{1,\delta+1}} + \Vert DX_k\Vert_{W^p_{1,\delta+1}}+\Vert U\Vert_{W^p_{1,\delta+1}} \right)\cdot\\
&\;\;\;\;\;\Vert D(X-X_k)\Vert_{C^0_{\delta'+1+\frac{n}{p}}},
\end{align*}
thus, since $\delta'>-{\frac{n}{p}}$, then  $r^{-(\delta'+\frac{n}{p})}r^{-\frac{n}{p}}\in L^p$, which, since $\{X_k\}\subset W^p_{2,\delta}$ is supposed to be bounded, implies that 
\begin{align*}
\!\!\big\Vert \vert\tilde{K}(Y)\vert^2_{\gamma}-\vert\tilde{K}(Y_k)\vert^2_{\gamma}\big\Vert_{L^p_{\delta+2}}&\lesssim \left( \Vert DX\Vert_{W^p_{1,\delta+1}} + \Vert DX_k\Vert_{W^p_{1,\delta+1}}+\Vert U\Vert_{W^p_{1,\delta+1}} \right)\Vert D(X-X_k)\Vert_{C^0_{\delta'+1+\frac{n}{p}}},
\end{align*}
where the right-hand side goes to zero by hypothesis. The same line of reasoning proves the analogous statement for $\epsilon_2$, and the coefficients which are independent of $Y$ trivially satisfy this property. We also need to analyse the coefficient $\omega_2(Y)=\tilde{F}(Df,\cdot) + \tilde{F}(\vartheta,\cdot)$. Clearly under our hypotheses $\omega_2\in L^p_{\delta+2}$ and also
\begin{align*}
\vert\omega_2(Y)-\omega_2(Y_k)\vert\lesssim r^{-(\delta+2+\frac{n}{p})}r^{-(\delta'+\frac{n}{p})}\Vert\tilde{F}\Vert_{W^p_{1,\delta+1}}\Vert Df-Df_k\Vert_{C^0_{\delta'+1+\frac{n}{p}}}.
\end{align*}
Thus, using again the fact that $r^{-\frac{n}{p}}r^{-(\delta'+\frac{n}{p})}\in L^p$, we get that
\begin{align*}
\Vert\omega_2(Y)-\omega_2(Y_k)\Vert_{L^p_{\delta+2}}\lesssim \Vert\tilde{F}\Vert_{W^p_{1,\delta+1}}\Vert Df-Df_k\Vert_{C^0_{\delta'+1+\frac{n}{p}}}\rightarrow 0.
\end{align*}
Finally, concerning the $b^{\alpha}_J$-coefficients associated with the boundary conditions, since we have already noticed that none of these coefficients depend on $f$ or $X$, we conclude that compactness property holds, which establishes the lemma.
\end{proof}

In order to get existence results we will need both isomorphism theorems for the linear parts associated to systems of the form of (\ref{generalizedsystem}) and typical elliptic estimates, which are key ingredient to prove Fredholm properties. More explicitly, we are considering operators $(L^i,B^i):W^p_{2,\delta}(M)\mapsto L^p_{\delta+2}(M)\times W^p_{1-\frac{1}{p}}(\partial M)$, with $p>n$ and $-\frac{n}{p}<\delta<n-2-\frac{n}{p}$, satisfying estimates of the form
\begin{align}\label{Fredholmestimate}
\Vert Y^i\Vert_{W^p_{2,\delta}}\leq C^i\Big(\Vert L^iY^i\Vert_{L^p_{\delta+2}(M)} + \Vert B^iY^i\Vert_{W^p_{1-\frac{1}{p}}(\partial M)} \Big).
\end{align}
For the Laplacian and conformal Killing Laplacian these properties were obtained in \cite{Maxwell1}. For compact manifolds with boundary these properties have been widely studied, and we can find a nice summary in \cite[Theorem 1.6.2]{Schwartz}. Putting this together with results for complete AE manifolds, such as those of \cite{Cantor,CB-C,Bartnik}, we can get a very general class of elliptic second order operators on AE manifolds satisfying (\ref{Fredholmestimate}). All this motivates the following definition. 
\begin{defn}\label{einsteintypesystems}
We will say that a system of the form of (\ref{generalizedsystem}) is a \textbf{conformal Einstein-type system} if its coefficients satisfy the mapping properties described below (\ref{generalizedsystem}), together with the boundedness and compactness properties (\ref{boundednessprop})-(\ref{compactnessprop}) and the Fredholm estimate (\ref{Fredholmestimate}). 
\end{defn}
In this context we will still denote the linear part associated to the generalized system (\ref{generalizedsystem}) by $\mathcal{P}$ and the non-linearities appearing in the right-hand side by $\textbf{F}$, and, furthermore, we will associate the shifted operators $\mathcal{P}_{a,b}$
\begin{align*}
\mathcal{P}_{a,b}:W^{p}_{2,\delta}(M;E)&\mapsto L^p_{\delta+2}(M;E)\times W^p_{1-\frac{1}{p}}(\partial M;E),\\
(\phi,Y)&\mapsto (\Delta_{\gamma}\varphi-a\varphi,L^i(Y^i),-(\hat{\nu}(\varphi)+b\varphi)\vert_{\partial M},B^iY^i\vert_{\partial M}).
\end{align*}
for any choice of $a\in L^p_{\delta+2}(M)$, $b\in W^p_{1-\frac{1}{p}}(\partial M)$ satisfying $a,b\geq 0$. 
\begin{thm}\label{conformalexistence}
Let $(M^n,\gamma)$ be a $W^p_{2,\delta}$-AE manifold, with $p>n$ and $-\frac{n}{p}<\delta<n-2-\frac{n}{p}$, and consider a conformal Einstein-type system of the form of (\ref{generalizedsystem}) on $M$. Assume that the Lichnerowicz equation admits a compatible pair of strong global sub and supersolutions given by $\phi_{-}$ and $\phi_{+}$, which are, respectively, asymptotic to harmonic functions $\omega_{\pm}$ tending to positive constants $\{A^{\pm}_{j} \}_{j=1}^N$ on each end $\{E_j \}_{j=1}^{N}$. Fix a harmonic function $\omega$ asymptotic to constants $\{A_j \}_{j=1}^{N}$ on each end satisfying $0<A^{-}_j\leq A_j\leq A^{+}_j$, and suppose that the map 
\begin{align*}
\mathcal{F}_{a,b}:W^p_{2,\delta}(M;E)&\mapsto W^p_{2,\delta}(M;E),\\
\psi=(\varphi,Y)&\mapsto \mathcal{F}_{a,b}(\psi)\doteq \mathcal{P}^{-1}_{a,b}\circ\textbf{F}_{a,b}(\psi).
\end{align*}
is actually invariant on the set $\times_{i}B_{M_{Y^i}}\subset \oplus_{i}W^p_{2,\delta}(M,E_i)$ given in the definition of the barriers $\phi_{-},\phi_{+}$ for any $\phi_{-}\leq\varphi + \omega \leq \phi_{+}$. Then, the system admits a solution $(\phi=\omega+\varphi,Y)$, with $(\varphi,Y)\in W^p_{2,\delta}(M;E)$, and, furthermore, $\phi>0$. 
\end{thm}

\begin{proof}
First of all, consider the shifted system associated with (\ref{generalizedsystem}), where we will pick the shift functions $a\in L^p_{\delta+2}$ and $b\in W^p_{1-\frac{1}{p}}$ below. We have our strong global sub and supersolutions fixed together with the balls $B_{M_{Y^i}}\subset W^p_{2,\delta}(M,E_i)$. From our hypotheses on the barriers $\phi_{\pm}$, we know that these are bounded functions, which implies that there are positive numbers, say $l\leq m$, such that $l\leq\phi_{-}\leq\phi_{+}\leq m$. Thus, consider  $Y\in \times_{i}B_{M_{Y^{i}}}$, and define the functions
\begin{align*}
h^{a}_{Y}(\phi)&\doteq h_{Y}(\phi) - a(\phi-\omega)=\sum_{I}a^0_I(Y)\phi^{I} - a(\phi-\omega),\\
g^{b}_{Y}(\phi)&\doteq g_{Y}(\phi) - b(\phi-\omega)=\sum_{K}b^0_K(Y)\phi^{K} - b(\phi-\omega) \text{ on } \partial M.
\end{align*}
We want to pick the functions $a$ and $b$ such that both $h^a_{Y}(y),g^b_Y(y)$ are decreasing functions on $y\in [l,m]$, for all $Y\in \times_iB_{M_{Y^i}}$. Thus, notice that
\begin{align*}
\frac{\partial}{\partial y}\sum_{I}a^0_I(Y)y^{I}&\leq \sum_{I}\vert I\vert \vert a^0_I(Y)\vert y^{I-1}\\
&\leq \sum_{I}\vert I\vert f_Iy^{I-1}\leq  \sum_{I}\vert I\vert\sup_{l\leq y\leq m}y^{I-1} \;f_I \in L^p_{\delta+2}, \; \forall \; Y\in B_{M_Y}
\end{align*}
where we used the boundedness property (\ref{boundednessprop}). With similar arguments, we get that
\begin{align*}
\frac{\partial}{\partial y}\sum_{K}b^0_K(Y)y^{K}&\leq \sup_{l\leq y\leq m}\sum_{K}\vert K\vert \vert b^0_K(Y)\vert y^{K-1} \leq \sup_{l\leq y\leq m}\sum_{K}\vert K\vert g_{K}y^{K-1},\\
&\leq \sum_{K}\vert K\vert\sup_{l\leq y\leq m}y^{K-1} \; g_{K}\in W^p_{1-\frac{1}{p}}, \; \forall \; Y\in B_{M_{Y}},
\end{align*}
where we have again used the boundedness property (\ref{boundednessprop}). Thus, if we pick $a\in L^p_{\delta+2}(M)$ and $b\in W^p_{1-\frac{1}{p}}(\partial M)$ satisfying
\begin{align}\label{shiftchoice}
\begin{split}
a> \sum_{I}\vert I\vert\sup_{l\leq y\leq m}y^{I-1} \;f_I \;\; ; \;\; b>\sum_{K}\vert K\vert\sup_{l\leq y\leq m}y^{K-1} \; g_{K},
\end{split}
\end{align}
we get that
\begin{align*}
\frac{\partial}{\partial y}h^a_{Y}(y)\leq 0, \text{ and } \frac{\partial}{\partial y}g^b_{Y}(y)\leq 0, \text{ for all } (y,Y)\in [l,m]\times_{i}B_{M_{Y^i}},
\end{align*}
implying that $h^a_{Y}(\phi)$ and $h^b_{Y}(\phi)$ are decreasing functions on the interval $[\phi_{-},\phi_{+}]_{C^0}$.

Now, consider $\psi_0=(\phi_{0},Y_{0})$ with $\phi_{0}=\phi_{-}$ and $Y_0\in \times_iB_{M_{Y^i}}$, and consider the sequence $\{\psi_k=(\phi_{k}=\omega + \varphi_k,Y_k)\}_{k=1}^{\infty}$ defined by an iteration procedure of the form:
\begin{align*}
\mathcal{P}_{a,b}\psi_k=\textbf{F}_{a,b}(\psi_{k-1}).
\end{align*}
From the linear properties associated with the operator $\mathcal{P}_{a,b}$ we know that the sequence is well-defined, since for each step $\textbf{F}_{a,b}(\psi_{k-1})\in L^p_{\delta+2}$. Furthermore, from our hypotheses, we know that $Y_k\in \times_iB_{M_{Y^i}}$ for all $k$ as long as we guarantee that $\phi_{k}=\omega + \varphi_k$ stays in the interval $[\phi_{-},\phi_{+}]_{C^0}$. We can prove this last statement inductively as follows. First consider $\phi_1$, which satisfies
\begin{align*}
\begin{pmatrix} 
\Delta_{\gamma}(\phi_1-\phi_{-}) - a ((\phi_1-\omega)-(\phi_{-}-\omega)) \\
-\big(\hat{\nu}(\phi_1-\phi_{-}) + b(\phi_1-\omega-(\phi_{-}-\omega))\big)
\end{pmatrix}
&=-\begin{pmatrix}
\Delta_{\gamma}\phi_{-} - h_{Y_0}(\phi_{-}) \\
-\hat{\nu}(\phi_{-}) - g_{Y_0}(\phi_{-})
\end{pmatrix}\leq 0.
\end{align*}
where the final inequality is a consequence of $\phi_{-}$ being a strong global subsolution. Notice that $\phi_{-}$ is in $W^p_{2,loc}$ and is asymptotic to $\omega_{-}$, which itself tends to positive constants $\{A_j^{-}\}_{j=1}^N$ in each end $\{E_j\}_{j=1}^N$ respectively, and, by construction, $\phi_1$ is asymptotic to $\omega$ which tends to positive constants $\{A_j\}_{j=1}^N$ in each end, satisfying $A_j>A_j^{-}$. Thus, we get that $\phi_{1}-\phi_{-}\rightarrow A_j-A^{-}_j>0$ in each end $E_j$.  Then, we can apply the weak maximum principle given in \cite[Lemma A.1]{Holst2} and conclude that $\phi_{1}\geq\phi_{-}$. Similarly, we have that
\begin{align*}
\!\!\!\!\!\!\!\!\begin{pmatrix}
\Delta_{\gamma}(\phi_{+}-\phi_{1}) - a ((\phi_{+}-\omega)-(\phi_1-\omega))\\
-\hat{\nu}(\phi_{+}-\phi_{1}) - b((\phi_{+}-\omega)-(\phi_1-\omega))
\end{pmatrix}
&=\begin{pmatrix}
\Delta_{\gamma}\phi_{+}- h_{Y_0}(\phi_{+}) +  h^a_{Y_0}(\phi_{+}) - h^a_{Y_0}(\phi_{-})\\
-\hat{\nu}(\phi_{+}) - g_{Y_0}(\phi_{+})+ g^b_{Y_0}(\phi_{+}) - g^b_{Y_0}(\phi_{-})
\end{pmatrix}\\
&\leq 0
\end{align*}
where, again, the last inequality follows form $\phi_{+}$ being a strong global supersolution, and also $h^a_{Y}(\cdot)$ and $g^a_{Y}(\cdot)$ being a decreasing functions of $\phi\in [\phi_{-},\phi_{+}]_{C^0}$ for any $Y\in \times_iB_{M_{Y^i}}$. Then, noticing again that $\phi_{+}$ is asymptotic to $\omega_{+}$, which tends to constants $\{A^{+}_j\}_{j=1}^N$ in each end, satisfying $A^{+}_j>A_j$, we get that $\phi_{+}-\phi_1\xrightarrow[]{E_j} A^{+}_j-A_j>0$, which implies through the maximum principle that $\phi_{+}\geq\phi_{1}$.

Now, assuming $\phi_{-}\leq\phi_{n}=\omega + \varphi_n\leq\phi_{+}$ and $Y_n\in \times_{i}B_{M_{Y^{i}}}$, then 
\begin{align*}
\!\!\!\!\begin{pmatrix}
\Delta_{\gamma}(\phi_{n+1}-\phi_{-}) - a ((\phi_{n+1}-\omega)-(\phi_{-}-\omega))\\
-\hat{\nu}(\phi_{n+1}-\phi_{-})-b((\phi_{n+1}-\omega)-(\phi_{-}-\omega))
\end{pmatrix}
&=-
\begin{pmatrix}
\Delta_{\gamma}\phi_{-} - h_{Y_n}(\phi_{-}) + h^a_{Y_n}(\phi_{-}) - h^a_{Y_n}(\phi_{n}) \\
- \hat{\nu}(\phi_{-}) - g_{Y_n}(\phi_{-}) +  g^b_{Y_n}(\phi_{-}) - g^b_{Y_n}(\phi_n)
\end{pmatrix}\\
&\leq 0,
\end{align*} 
where the last inequality holds since $\phi_{-}$ is a strong global subsolution. Furthermore, $h^a_{Y}$ and $g^b_{Y}$ are decreasing functions of $\phi\in [\phi_{-},\phi_{+}]_{C^0}$ for any $Y\in \times_{i}B_{M_{Y^i}}$, and from the inductive hypothesis $\phi_{-}\leq \phi_n\leq\phi_{+}$ and $Y_n\in \times_{i}B_{M_{Y^i}}$. Thus, since $\phi_{n+1}-\phi_{-}\xrightarrow[]{E_j} A_j-A^{-}_j>0$, we get $\phi_{n+1}\geq\phi_{-}$. Similarly, 
\begin{align*}
\!\!\!\!\begin{pmatrix}
\Delta_{\gamma}(\phi_{+}-\phi_{n+1}) - a ((\phi_{+}-\omega)-(\phi_{n+1}-\omega))\\
-\hat{\nu}(\phi_{+}-\phi_{n+1}) - b ((\phi_{+}-\omega)-(\phi_{n+1}-\omega))
\end{pmatrix}
&= 
\begin{pmatrix}
\Delta_{\gamma}\phi_{+} - h_{Y_n}(\phi_{+}) + h^a_{Y_n}(\phi_{+}) - h^a_{Y_n}(\phi_{n}) \\
-\hat{\nu}(\phi_{+})- g_{Y_{n}}(\phi_{+}) +  g^b_{Y_{n}}(\phi_{+}) - g^b_{Y_{n}}(\phi_n)
\end{pmatrix}\\
& \leq 0,
\end{align*}
where the last inequality holds because $\phi_{+}$ is a strong global supersolution; $h^a_Y(\cdot)$ and $g^b_Y(\cdot)$ are decreasing function of $\phi\in [\phi_{-},\phi_{+}]_{C^0}$ for all $Y\in \times_{i}B_{M_{Y^i}}$ and the inductive hypotheses. Thus, again, the maximum principle implies that $\phi_{+}\geq\phi_{n+1}$. All this implies that $\phi_{-}\leq\phi_{n+1}\leq\phi_{+}$, which finishes the inductive proof. Hence we have produced the sequence of solutions $\{(\varphi_n,Y_n) \}_{n=0}^{\infty}\subset W^p_{2,\delta}(M;E)$, where $Y_n\in \times_i B_{M_{Y^i}}$ and $\phi_{n}=\varphi_n + \omega\in [\phi_{-},\phi_{+}]_{C^0}$ for all $n$. Notice that this implies
\begin{align*}
\Vert\varphi_n\Vert_{W^p_{2,\delta}}&\lesssim 
\sum_{I}\Vert a^0_I(Y_{n-1})\Vert_{L^p_{\delta+2}}\Vert\phi_{n-1}^I\Vert_{C^0} + \Vert a\Vert_{L^p_{\delta+2}}\Vert \phi_{n-1}\Vert_{C^0} \\
&+  \sum_{K}\Vert b^0_K(Y_{n-1})\phi_{n-1}^K\Vert_{W^p_{1-\frac{1}{p}}} + \Vert b\phi_{n-1}\Vert_{W^p_{1-\frac{1}{p}}},\\
&\lesssim \sum_{K}\Vert b^{0}_K(Y_{n-1})\phi_{n-1}^K\Vert_{W^p_{1-\frac{1}{p}}} + \Vert b\phi_{n-1}\Vert_{W^p_{1-\frac{1}{p}}} + \sum_{I}\Vert f_I\Vert_{L^p_{\delta+2}}\Vert\phi_{\pm}^I\Vert_{C^0} + \Vert a\Vert_{L^p_{\delta+2}}\Vert \phi_{+}\Vert_{C^0},
\end{align*}
Notice that since $b^{0}_K(Y_{n-1})\in W^p_{1-\frac{1}{p}}(\partial M)$ for any $n$ and any $K$, then we know that there are (non-unique) extensions $\tilde{b}^{0}_K(Y_{n-1})\in W^p_{1}(U)$, where $U$ is some smooth neighbourhood of $\partial M$ with compact closure, and we also know that both the extension operator and the trace map are continuous. Therefore, $\Vert b^{0}_K(Y_{n-1})\Vert_{W^p_{1-\frac{1}{p}}(\partial M)}\lesssim \Vert\tilde{b}^{0}_K(Y_{n-1})\Vert_{W^p_{1}(M)}$, and, clearly, $\tilde{b}^{0}_K(Y_{n-1})\phi^{K}\in W^p_{1}(U)$ is and extension of $b^{0}_K(Y_{n-1})\phi^{K}\vert_{\partial M}$. Therefore, using the fact that $W^p_1(U)$ is an algebra under multiplication for $p>n$, we see that 
\begin{align*}
\Vert b^{0}_K(Y_{n-1})\phi^{K}_{n-1}\Vert_{W^p_{1-\frac{1}{p}}(\partial M)}&\lesssim \Vert\tilde{b}^{0}_K(Y_{n-1})\phi^{K}_{n-1}\Vert_{W^p_1(U)}\lesssim  \Vert\tilde{b}^{0}_K(Y_{n-1})\Vert_{W^p_1(U)}\Vert\phi^{K}_{n-1}\Vert_{W^p_1(U)}.
\end{align*}
Notice that 
\begin{align*}
\Vert\phi^{K}_{n-1}\Vert_{W^p_1(U)}&\lesssim \Vert\phi^{K}_{n-1}\Vert_{L^p(U)} + \Vert\nabla \phi^{K}_{n-1}\Vert_{L^p(U)}\lesssim \Vert\phi^{K}_{\pm}\Vert_{C^0} + \Vert \phi^{K-1}_{\pm}\Vert_{C^0} \Vert\nabla\phi_{n-1}\Vert_{L^p(U)},\\
&\lesssim 1 + \Vert\phi_{n-1}\Vert_{W^p_1(U)}\leq 1 + \Vert\omega\Vert_{W^p_1(U)} + \Vert\varphi_{n-1}\Vert_{W^p_1(U)},
\end{align*}
where in the first line $\phi_{\pm}$ stands for the subsolution if the exponents are negative and the supersolution if they are positive, and the implicit constant depends on the barriers, but not on $n$. Using the above and interpolation inequalities, we get that for any $\epsilon>0$ it holds that 
\begin{align*}
\Vert b^{0}_K(Y_{n-1})\phi^{K}_{n-1}\Vert_{W^p_{1-\frac{1}{p}}(\partial M)}&\lesssim \Vert\tilde{b}^{0}_K(Y_{n-1})\Vert_{W^p_1(U)}(1+  	\Vert\omega\Vert_{C^1(U)} + \Vert\varphi_{n-1}\Vert_{W^p_1(U)})\\
&\lesssim \Vert\tilde{b}^{0}_K(Y_{n-1})\Vert_{W^p_1(U)} \\
&+ \Vert\tilde{b}^{0}_K(Y_{n-1})\Vert_{W^p_1(U)}(\epsilon\Vert\varphi_{n-1}\Vert_{W^p_2(U)} + C_{\epsilon}\Vert\varphi_{n-1}\Vert_{L^p(U)}),\\
&\lesssim C_K(1 + \epsilon\Vert\varphi_{n-1}\Vert_{W^p_{2,\delta}(M)} + C_{\epsilon}\Vert\varphi_{n-1}\Vert_{C^0(U)})
\end{align*}
Since $\phi_{-}-\omega\leq \varphi_n \leq \phi_{+} - \omega$  for all $n$, then $\Vert\varphi_{n-1}\Vert_{C^0(U)}\leq C_{\pm}$ for some fixed constant independent of $n$, which depends on the barriers and $\omega$. Therefore, we see that
\begin{align*}
\Vert b^{0}_K(Y_{n-1})\phi^{K}_{n-1}\Vert_{W^p_{1-\frac{1}{p}}(\partial M)}&\lesssim C_K(1  + C_{\epsilon}C_{\pm} + \epsilon\Vert\varphi_{n-1}\Vert_{W^p_{2,\delta}(M)}),
\end{align*}
which implies that
\begin{align*}
\!\!\!\!\!\!\!\!\!\Vert\varphi_n\Vert_{W^p_{2,\delta}}&\lesssim  \sum_{K}C_K(1  + C_{\epsilon}C_{\pm} + \epsilon\Vert\varphi_{n-1}\Vert_{W^p_{2,\delta}(M)})+ \Vert b\Vert_{W^p_{1-\frac{1}{p}}}(1  + C_{\epsilon}C_{\pm} + \epsilon\Vert\varphi_{n-1}\Vert_{W^p_{2,\delta}(M)}) \\
&+ \sum_{I}\Vert f_I\Vert_{L^p_{\delta+2}}\Vert\phi_{\pm}^I\Vert_{C^0} + \Vert a\Vert_{L^p_{\delta+2}}\Vert\phi_{+}\Vert_{C^0}.
\end{align*}
Now, choosing $\epsilon$ sufficiently small, we can write
\begin{align}
\Vert\varphi_n\Vert_{W^p_{2,\delta}}&\leq  \frac{1}{2}\Vert\varphi_{n-1}\Vert_{W^p_{2,\delta}} + C,
\end{align}
where $C$ is a fixed constant that only depends on the parameters of the problem, that is, it depends on the barriers, the functions $f_I$, the shift functions $a$ and $b$, the constants $C_K$ and the choice of $\epsilon$. We can iterate the above procedure to get that
\begin{align}
\Vert\varphi_n\Vert_{W^p_{2,\delta}}&\leq  \frac{1}{2^n}\Vert\varphi_{-}\Vert_{W^p_{2,\delta}} + C\sum_{j=0}^{n-1} 2^{-j}\leq \Vert\varphi_{-}\Vert_{W^p_{2,\delta}} + 2C \;\; \forall \;\; n.
\end{align}
The above estimate implies that there is a constant $M_{\varphi}>0$, depending on the barriers $\phi_{\pm}$, the norms of the functions $f_I$ and the constants $C_K$, such that $\{\varphi_n \}_{n=0}^\infty\subset B_{M_{\varphi}}\subset W^p_{2,\delta}(M;\mathbb{R})$, where $B_{M_{\varphi}}$ stands for the closed ball of radius $M_{\varphi}$. This implies that the sequence $\{(\varphi_n,Y_n) \}_{n=0}^{\infty}\subset B_{M_{\varphi}}\times_iB_{M_{Y^i}}$. Thus, since the embedding $W^p_{2,\delta}\hookrightarrow C^1_{\delta'+\frac{n}{p}}$ is compact for any $-\frac{n}{p}<\delta'<\delta$ (see Lemma \ref{compactembed1} and Corollary \ref{compactembed2}), we get that, up to restricting to a subsequence, 
\begin{align*}
(\varphi_n,Y_n)\xrightarrow[n\rightarrow\infty]{} (\varphi,Y) \text{ in } C^1_{\delta'+\frac{n}{p}}.
\end{align*}
Thus, using our elliptic estimates, notice that
\begin{align*}
\Vert(\varphi_n,Y_n)-(\varphi_m,Y_m)\Vert_{W^p_{2,\delta}}
&\lesssim \sum_I\Vert a^0_{I}(Y_{n-1})\phi_{n-1}^{I}-a^0_{I}(Y_{m-1})\phi_{m-1}^{I}\Vert_{L^p_{\delta+2}} \\
&+ \Vert a\Vert_{L^p_{\delta+2}}\Vert\phi_{n-1}-\phi_{m-1}\Vert_{C^0}\\
&+ \sum_K\Vert b^0_{K}(Y_{n-1})\phi_{n-1}^{K}-b^0_{K}(Y_{m-1})\phi_{m-1}^{K}\Vert_{W^p_{1-\frac{1}{p}}} \\
&+ \Vert b(\phi_{n-1}-\phi_{m-1})\Vert_{W^p_{1-\frac{1}{p}}}\\
&+\sum_i\Big\{\sum_J\Vert a^i_{J}(Y_{n-1})\phi_{n-1}^{J}-a^i_{J}(Y_{m-1})\phi_{m-1}^{J}\Vert_{L^p_{\delta+2}} \\
&+  \sum_L\Vert b^i_{L}(Y_{n-1})\phi_{n-1}^{L}-b^i_{L}(Y_{m-1})\phi_{m-1}^{L}\Vert_{W^p_{1-\frac{1}{p}}}\Big\}.
\end{align*}
Now, consider the following estimates
\begin{align*}
\Vert a^0_{I}(Y_{n-1})\phi_{n-1}^{I}-a^0_{I}(Y_{m-1})\phi_{m-1}^{I}\Vert_{L^p_{\delta+2}}&\leq \Vert a^0_{I}(Y_{n-1}) - a^0_{I}(Y_{m-1})\Vert_{L^p_{\delta+2}}\Vert\phi_{n-1}^{I}\Vert_{C^0} \\
&+ \Vert a^0_{I}(Y_{m-1})\Vert_{L^p_{\delta+2}}\Vert\phi_{n-1}^{I} - \phi_{m-1}^{I}\Vert_{C^0}. 
\end{align*}
Therefore, from the compactness property and the boundedness of $\{Y_n \}$ in $W^p_{2,\delta}$, we know that $\{a^{0}_I(Y_n)\}_{n=0}^{\infty}$ is Cauchy $L^p_{\delta+2}$. In particular this implies that this sequence converges and thus it is bounded. We also have that $C^1_{\delta'+\frac{n}{p}}\hookrightarrow C^0$, since $\delta'>-\frac{n}{p}$, which implies $C^0$-convergence of $\phi_{n}\rightarrow \phi $, we get that
\begin{align*}
\Vert a^{0}_{I}(Y_{n-1})\phi_{n-1}^{I}-a^{0}_{I}(Y_{m-1})\phi_{m-1}^{I}\Vert_{L^p_{\delta+2}}\xrightarrow[n,m\rightarrow\infty]{} 0,
\end{align*}
and a similar statement holds for the terms involving $a^i_{J}$. Concerning the coefficients $b^{0}_K$, similarly to what we did previously, we have that
\begin{align*}
\!\!\!\!\!\!\Vert b^{0}_{K}(Y_{n-1})\phi_{n-1}^{K} - b^{0}_{K}(Y_{m-1})\phi_{m-1}^{K}\Vert_{W^p_{1-\frac{1}{p}}} &\leq \Vert(b^0_{K}(Y_{n-1}) - b^0_{K}(Y_{m-1}))\phi_{n-1}^{K}\Vert_{W^p_{1-\frac{1}{p}}} \\
&+ \Vert b^0_{K}(Y_{m-1})(\phi_{n-1}^{K} - \phi_{m-1}^{K})\Vert_{W^p_{1-\frac{1}{p}}},\\
&\lesssim \Vert (b^0_{K}(Y_{n-1}) - b^0_{K}(Y_{m-1}))\Vert_{W^p_{1-\frac{1}{p}}}\Vert\phi_{n-1}^{K}\Vert_{W^p_{1}(U)} \\
&+ \Vert b^0_{K}(Y_{m-1})\Vert_{W^p_{1-\frac{1}{p}}}\Vert \phi_{n-1}^{K} - \phi_{m-1}^{K}\Vert_{W^p_{1}(U)},
\end{align*}
which goes to zero because of the compactness hypotheses for the coefficients and the fact that $\varphi_n\xrightarrow[]{} \varphi$ in $C^1_{\delta'+\frac{n}{p}}$ implies that $\phi_n\xrightarrow[]{} \phi$ in $C^1$ on compact sets. A similar thing occurs with the $b^{i}_L$-terms. Finally, the remaining terms also go to zero as $n$ and $m$ go to infinity because $a\in L^p_{\delta+2}$, $b\in W^p_{1-\frac{1}{p}}$, $\phi_n\xrightarrow[]{}\phi$ in $C^0$ and $\phi_n\xrightarrow[]{} \phi$ in $C^1$ on compact sets. All this implies that $\{(\phi_n,Y_n)\}_{n=0}^{\infty}$ is actually Cauchy in $W^p_{2,\delta}$, which improves the convergence and finishes the proof.
 
\end{proof}


Putting together the Theorem \ref{conformalexistence} with Lemma \ref{EMcompactness} we obtain an existence criteria for the charged fluid constraints (\ref{Conformal-EMSystem.1})-(\ref{boundcondsystems})-(\ref{ElectricPotentialEq}). Thus, we see that what we need to do is to construct strong global barriers for the constraint system, and to prove invariance of the map $\mathcal{F}_{a,b}(\varphi,\cdot,\cdot)$ on the balls $B_{M_f}\times B_{M_Y}$ for any $\varphi\in [\varphi_{-},\varphi_{+}]_{C^0}$. To achieve this, we will use some explicit elliptic estimates associated to the Einstein-Maxwell constraint system.

Let us highlight that the invariance property of the map $\mathcal{F}_{a,b}$ associated to the shifted system is actually independent of the shift functions $a$ and $b$, in the sense that if this property holds for the map $\mathcal{F}$, then it holds for $\mathcal{F}_{a,b}$ as well. This can be seen as follows. Since the shifts only affect the Lichnerowicz equation, if $(\phi,Y)=\mathcal{F}_{a,b}(\bar{\phi},\bar{Y})$, then $Y^i={\mathcal{P}^{-1}}^i(\bar{\phi},\bar{Y})$ and furthermore $\mathcal{F}(\bar{\phi},\bar{Y})=(\hat{\phi},Y)$, where only the first component is affected by the presence of the shifts. Thus, if there are radii $M_{Y^i}$ such that the map $\mathcal{F}$ is invariant on $\times_iB_{M_i}$ for any $\phi_{-}\leq \phi\leq \phi_{+}$, then this means that $\Vert Y^{i}\Vert_{W^{p}_{2,\delta}}\leq M_{Y^i}$ for any $(\bar{\phi},\bar{Y})\in [\phi_{-},\phi_{+}]_{C^{0}}\times_iB_{M_i}$, proving the invariance property of $\mathcal{F}_{a,b}$ stated in the above theorem. Therefore, this property can be recast as in Theorem \ref{ExistenceMetathmInto}. Furthermore, in practice, the radii $M_{Y^i}$ of the corresponding balls will be fixed via the elliptic estimates (\ref{Fredholmestimate}) applied to such a solution, estimating $\bar{\phi}$ in the right-hand side of such an inequality through the barriers. Estimating the coefficients depending of $(Y,DY)$ relies on the structure of each specific system. As we shall see below, the electromagnetic system can be done by first fixing the radius $M_f$ following this recipe and exploiting that the coefficients in (\ref{ElectricPotentialEq}) are independent of $X$. Then, one fixes $M_X$ through a priori estimates as above, where the coefficients involving $(f,Df)$ will be estimated appealing to $M_f$. This example will show that this kind of \emph{triangular systems} are particularly well-suited for applications of the above theorem.

\subsubsection*{Electromagnetic constraint}

Within the kind of iteration scheme described above, we need to get a global fixed estimates for solutions of linear equations of the form
\begin{align}\label{electriciteration}
\begin{split}
\Delta_{\gamma}f &= \tilde{q}\bar{\phi}^{\frac{2n}{n-2}},\\
-\hat{\nu}(f)&=E_{\hat{\nu}} \text{ on } \partial M,
\end{split}
\end{align}
where, in the right-hand side, the functions $\bar{\phi}=\omega+\varphi$, $\varphi\in W^p_{2,\delta}(M), \tilde{q}\in L^p_{\delta+2}(M)$ and $E_{\hat{\nu}}\in W^p_{1-\frac{1}{p}}(\partial M)$, with $p>n$ and $-\frac{n}{p}<\delta<n-2-\frac{n}{p}$, are considered as a given data. Thus, the right hand is in $L^p_{\delta+2}(M)\times W^p_{1-\frac{1}{p}}(\partial M)$, and therefore any $W^p_{2,\delta}$-solution of (\ref{electriciteration}) satisfies the the following elliptic estimate:
\begin{align}\label{f-estimate}
\Vert f\Vert_{W^p_{2,\delta}}\leq C \Big\{\Vert\tilde{q}\Vert_{L^p_{\delta+2}}\Vert\bar{\phi}^{\frac{2n}{n-2}}\Vert_{C^0} + \Vert E_{\hat{\nu}}\Vert_{W^p_{1-\frac{1}{p}}} \Big\}.
\end{align}
In particular, in the above estimate we can get rid of the dependence on the specific $\bar{\phi}$ by admitting the existence of global supersolution.

\subsubsection*{Momentum constraint}

Similarly to what we did above, we want to get uniform estimates on the sequence of solutions generated by the momentum constraint, given by:
\begin{align}\label{MomentumConstraintEstimate0}
\begin{split}
\Delta_{\gamma,conf}X&=r_nD\tau \phi^{\frac{2n}{n-2}} + \omega_1\phi^{2\frac{n + 1}{n-2}} - \omega_2,\\
\pounds_{\gamma,conf}X(\hat{\nu},\cdot)&=-\Big(\big( \frac{1}{2}\vert\theta_{-}\vert - r_{n}\tau\big)v^{\frac{2n}{n-2}} + U(\hat{\nu},\hat{\nu}) \Big)\hat{\nu} \text{ on } \partial M,
\end{split}
\end{align}
where ${\omega_1}_k=\mu\left( 1 + \vert\tilde{u}\vert^2_{\gamma} \right)^{\frac{1}{2}}\tilde{u}_k $ and ${\omega_2}_{k}=\tilde{F}_{ik}\tilde{E}^i$. Let us suppose that $\tau, U\in W^p_{1,\delta+1}$ and that $\gamma$ being a $W^p_{2,\delta}$-AE metric are given, where, as above, $p>n$ and $-\frac{n}{p}<\delta<n-2-\frac{n}{p}$, and consider that $\phi=\omega+\varphi$, with $\varphi\in W^p_{2,\delta}$, and finally that $f\in W^p_{2,\delta}$, $\vartheta\in W^p_{1,\delta+1}$, $\theta_{-}\in W^p_{1-\frac{1}{p}}$ and $v\in W^p_{2-\frac{1}{p}}$. Also, suppose that $\mu\in L^p_{\delta+2}$ and $\tilde{F}\in W^p_{1,\delta+1}$, which implies that $\tilde{F}\otimes \tilde{E}\in W^p_{1,\delta+2}$, since $\tilde{E}=df + \vartheta \in W^p_{1,\delta+1}$. Then, since the right hand side of the above equation is in $L^p_{\delta+2}\times W^p_{1-\frac{1}{p}}$, following Proposition \ref{isomorphismthm}, we can associate a unique solution to it, say $X_{\phi,f}\in W^{p}_{2,\delta}$, and, using Proposition \ref{injectivityestimate}, we can estimate the $W^{p}_{2,\delta}$-norm of $X_{\phi,f}$ in terms of $\phi$, $f$ and the free data. That is,
\begin{align*}
\begin{split}
\Vert X_{\phi,f}\Vert_{W^{p}_{2,\delta}}
&\lesssim \Vert D\tau\Vert_{L^p_{\delta+2}} \Vert\phi^{\frac{2n}{n-2}}\Vert_{C^0} +  \Vert\omega_1\Vert_{L^p_{\delta+2}} \Vert\phi^{2\frac{n + 1}{n-2}}\Vert_{C^0} + \Vert{\omega_2}_f\Vert_{L^p_{\delta+2}},\\
&+ \Vert\frac{1}{2}\vert\theta_{-}\vert - r_{n}\tau\Vert_{W^p_{1-\frac{1}{p}}}\Vert v^{\frac{2n}{n-2}}\Vert_{C^0} + \Vert U\Vert_{W^p_{1-\frac{1}{p}}}
\end{split}
\end{align*}
where we have used that $W^p_{2,\delta}\hookrightarrow C^0$. Notice that $\Vert \omega_2\Vert_{L^p_{\delta+2}}\lesssim \Vert\tilde{F}\otimes \tilde{E}\Vert_{W^{p}_{1,\delta+2}}\lesssim \Vert\tilde{F}\Vert_{W^{p}_{1,\delta+2}}\Vert df\Vert_{W^{p}_{1,\delta+2}} + \Vert\tilde{F}\Vert_{W^{p}_{1,\delta+2}}\Vert\vartheta\Vert_{W^{p}_{1,\delta+2}}$, which explicitly gives us that
\begin{align}\label{Y-estimate}
\begin{split}
\Vert X_{\phi,f}\Vert_{W^{p}_{2,\delta}}&\leq \kappa\Big\{\Vert D\tau\Vert_{L^p_{\delta+2}} \Vert\phi^{\frac{2n}{n-2}}\Vert_{C^0} +  \Vert\omega_1\Vert_{L^p_{\delta+2}} \Vert\phi^{2\frac{n + 1}{n-2}}\Vert_{C^0} + \Vert\tilde{F}\Vert_{W^p_{1,\delta+1}}\Vert df\Vert_{W^p_{1,\delta+1}} \\
&+ \Vert \tilde{F}\Vert_{W^p_{1,\delta+1}}\Vert\vartheta \Vert_{W^p_{1,\delta+1}} + \big\Vert \vert\theta_{-}\vert - 2r_{n}\tau\big\Vert_{W^p_{1-\frac{1}{p}}}\Vert v^{\frac{2n}{n-2}}\Vert_{C^0} + \Vert U\Vert_{W^p_{1-\frac{1}{p}}}\Big\}.
\end{split}
\end{align}

In case we have a pair of compatible strong global barriers $\phi_{-}\leq \phi_{+}$, then, there are radii $M_{f},M_{Y}\subset W^p_{2,\delta}$ such that Definition \ref{strongbarriers} works for the Lichnerowicz equation for any $f\in B_{M_f}$ and any $X\in B_{M_X}$. In such a case, notice that for any $f\in B_{M_f}$ and any $\phi_{-}\leq\phi\leq\phi_{+}$, we get that
\begin{align*}
\begin{split}
\Vert X_{\phi,f}&\Vert_{W^{p}_{2,\delta}}\lesssim \Vert D\tau\Vert_{L^p_{\delta+2}} \Vert\phi_{+}^{\frac{2n}{n-2}}\Vert_{C^0} +  \Vert\omega_1\Vert_{L^p_{\delta+2}} \Vert\phi_{+}^{2\frac{n + 1}{n-2}}\Vert_{C^0}\\
& + \Vert\tilde{F}\Vert_{W^p_{1,\delta+1}}\left(M_f  + \Vert\vartheta\Vert_{W^p_{1,\delta+1}}\right) + \Vert \vert\theta_{-}\vert - 2r_{n}\tau\Vert_{W^p_{1-\frac{1}{p}}}\Vert v^{\frac{2n}{n-2}}\Vert_{C^0} + \Vert U\Vert_{W^p_{1-\frac{1}{p}}},
\end{split}
\end{align*}

We shall need one further estimate related to the momentum constraint. Suppose that $X_{\phi,f}$ is a solution of (\ref{MomentumConstraintEstimate0}) for source functions $(\phi,f)$. Since $\vert\pounds_{\gamma,conf}X_{\phi,f}\vert_e\lesssim \vert DX_{\phi,f}\vert_e$, if $p>n$ and $X\in W^p_{2,\delta}$, we can appeal to Proposition \ref{Mazzeo3} to estimate $\vert DX_{\phi,f}\vert_e\lesssim r^{-(\delta+1+\frac{n}{2})}\Vert X_{\phi,f}\Vert_{W^p_{2,\delta}}$, implying that
\begin{multline}
\begin{split}
\!\!\!\!\!\!\vert\pounds&_{\gamma,conf}X_{\phi,f}\vert_e\lesssim r^{-(\delta+1+\frac{n}{2})}\Big\{ r_n\Vert D\tau\Vert_{L^p_{\delta+2}} \Vert\phi\Vert_{C^0}^{\frac{2n}{n-2}} +  \Vert\omega_1\Vert_{L^p_{\delta+2}} \Vert\phi\Vert_{C^0}^{{2\frac{n + 1}{n-2}}} + \Vert\tilde{q}\Vert_{L^p_{\delta+2}}\Vert\tilde{F}\Vert_{W^p_{1,\delta+1}}{\Vert\phi\Vert_{C^0}}^{\frac{2n}{n-2}} \\
& + \Vert E_{\hat{\nu}}\Vert_{W^p_{1-\frac{1}{p}}}\Vert\tilde{F}\Vert_{W^p_{1,\delta+1}} + \Vert\vartheta\Vert_{W^p_{1,\delta+1}}\Vert\tilde{F}\Vert_{W^p_{1,\delta+1}} + \big\Vert \vert\theta_{-}\vert - 2c_{n}\tau\big\Vert_{W^p_{1-\frac{1}{p}}}\Vert v^{\frac{2n}{n-2}}\Vert_{C^0} + \Vert U\Vert_{W^p_{1-\frac{1}{p}}} \big\}
\end{split}
\end{multline}


\subsubsection*{Barriers for the Lichnerowicz equation}

Consider the hamiltonian constraint, given by
\begin{align*}
\begin{split}
\!\!\!\!\!\!\!\!\!\!\!\Delta_{\gamma}\phi &= c_n R_{\gamma}\phi - c_n \vert\tilde{K}\vert^2_{\gamma}\phi^{-\frac{3n-2}{n-2}} - c_n\left(2\epsilon_1  - r_n\tau^2\right)\phi^{\frac{n+2}{n-2}} - 2c_n\epsilon_2\phi^{-3} - 2c_n\epsilon_3\phi^{\frac{n-6}{n-2}},\\
\!\!\!\!\!\!- \hat{\nu}(\phi)&=  a_n H \phi  - (d_n\tau + a_n\theta_{-})\phi^{\frac{n}{n-2}} - a_n\left( \frac{1}{2}\vert\theta_{-}\vert - r_{n}\tau\right)v^{\frac{2n}{n-2}}\phi^{-\frac{n}{n-2}}, \text{ on } \partial M
\end{split}
\end{align*}
According to the analysis made in the previous sections, we need to establish the following two properties:
\begin{itemize}
\item The above Lichnerowicz equation admits a compatible pair of strong global sub and super-solutions $0<\phi_{-}\leq\phi_{+}$, where $\phi_{\pm}=\omega_{\pm}+\varphi_{\pm}$ with $\varphi_{\pm}\in W^p_{2,\delta}$ and $\omega_{\pm}$ are as in Theorem \ref{conformalexistence}, and these barriers work for any $f\in B_{M_f}$ and $X\in B_{M_X}$ for suitable $M_f,M_Y>0$.
\item The map $\mathcal{F}_{a,b}:[\varphi_{-},\varphi_{+}]_{C^0}\times B_{M_f}\times B_{M_{X}}\mapsto W^p_{2,\delta}(M;E)$ associated to the shifted system must be invariant on $B_{M_f}\times B_{M_{X}}$ for any $\varphi\in [\varphi_{-},\varphi_{+}]_{C^0}$. 
\end{itemize}

The construction of the barriers will strongly depend on the Yamabe class of $\gamma$, which is something to be expected following previous works \cite{Maxwell1,Holst2}. Thus, let us recall that the Yamabe quotient related with an AE manifold $(M,g)$ with boundary $\partial M$ is given by the number:
\begin{align}
\mathcal{Q}_g(f)\doteq \frac{\int_M(\vert\nabla f\vert^2_g + c_nR_gf^2)dV + \int_{\partial M}a_nH_gf^2dA}{\Vert f\Vert^2_{L^{\frac{2n}{n-2}}}},\;\; \forall\; f\in C^{\infty}_{0}(M).
\end{align}
Then, we define the Yamabe invariant $\lambda_g$ as follows.
\begin{align}
\lambda_g\doteq \inf_{\underset{f\not\equiv 0}{f\in C^{\infty}_{0}}}\mathcal{Q}_g(f).
\end{align}

From \cite[Proposition 3 and Corollary 1]{Maxwell1}, we can extract the following important result concerning the Yamabe characterization in this context.

\begin{prop}
If $(M^n,g)$ is a $W^p_{2,\delta}$-Yamabe positive AE manifold with $p>\frac{n}{2}$, $-\frac{n}{p}<\delta<n-2-\frac{n}{p}$ and $n\geq 3$, then there is a conformal transformation $g'=\phi^{\frac{4}{n-2}}g$, with $\phi-1\in W^p_{2,\delta}$, such that $\lambda'_g>0$ making $(M,g')$ a scalar-flat $W^p_{2,\delta}$-AE manifold with zero boundary mean curvature.
\end{prop}


We can now construct strong global barriers for the Lichnerowicz equations associated with a charged fluid. Our constructions will be an adaptation of the barriers constructed in \cite{Holst2} to our present context, starting with far from CMC barriers for in the Yamabe positive case. Notice that the marginally trapped condition required that the data $\theta_{-}$ and $\tau$ satisfy $\frac{1}{2}\vert\theta_{-}\vert-r_n\tau\geq 0$. Then, notice that
\begin{align*}
a_n\vert\theta_{-}\vert - d_n\tau 
&=\frac{1}{2}\frac{n-2}{n-1}(\vert\theta_{-}\vert-\frac{n-1}{n}\tau)=a_n(\vert\theta_{-}\vert - r_n\tau)\geq a_n(\frac{1}{2}\vert\theta_{-}\vert - r_n\tau).
\end{align*}
Therefore, the (marginally) trapped condition $\frac{1}{2}\vert\theta_{-}\vert-r_n\tau\geq 0$ implies that $a_n\vert\theta_{-}\vert - d_n\tau\geq 0$.

\begin{lemma}\label{Yamabeposbarriers}
Let $(M,\gamma)$ be a $W^p_{2,\delta}$-Yamabe positive AE-manifold, with $p>n$, $n\geq 3$ and $-\frac{n}{p}<\delta<n-2-\frac{n}{p}$. Assuming the functional hypotheses given in Lemma \ref{EMcompactness}, if additionally $\mu\in W^{p}_{1,2\delta+1+\frac{n}{p}}$, $\left( \frac{1}{2}\vert\theta_{-}\vert - r_{n}\tau\right)\geq 0$ along $\partial M$ and $v>0$, then, under smallness assumptions on $\mu,\tilde{q},\tilde{F},E_{\hat{\nu}},\vartheta,v$ and $U$, the hamiltonian constraint associated to (\ref{Conformal-EMSystem.1})-(\ref{boundcondsystems}) admits compatible strong global barriers $0<\phi_{-}\leq\phi_{+}$, such that $\phi_{\pm}-\omega_{\pm}\in W^p_{2,\delta}$ for some harmonic functions $\omega_{\pm}$ which tend to positive constants $\{A_j^{\pm}\}_{j=1}^N$ on each end. Also, the map $\mathcal{F}_{a,b}$ associated to the system is invariant on balls $B_{M_f},B_{M_X}\subset W^p_{2,\delta}$ where the strong global barriers work for any $(a,b)\in L^p_{\delta+2}(M)\times W^{p}_{1-\frac{1}{p}}(\partial M)$ such that $a,b\geq 0$. Furthermore, under a smallness condition on $\big\Vert \vert\theta_{-}\vert - 2r_{n}\tau\big\Vert_{W^p_{1-\frac{1}{p}}}$, the choice $v\doteq \phi_{+}\vert_{\partial M}$ is compatible with this construction.
\end{lemma}
\begin{proof}
Let us begin by proving the existence of a strong global subsolution. We need to satisfy
\begin{align*}
\mathcal{H}^{1}_{f,X}(\phi_{-})&\doteq \Delta_{\gamma}\phi_{-} - c_n R_{\gamma}\phi_{-} - c_n\left( r_n\tau^2 - 2\epsilon_1\right)\phi_{-}^{\frac{n+2}{n-2}} + c_n \vert\tilde{K}(X)\vert^2_{\gamma}\phi_{-}^{-\frac{3n-2}{n-2}} \\
&+ 2c_n\epsilon_2(f)\phi_{-}^{-3} + 2c_n\epsilon_3\phi_{-}^{\frac{n-6}{n-2}}\geq 0, \\
\mathcal{H}^2_{f,X}(\phi_{-})&\doteq - \hat{\nu}(\phi_{-}) - a_n H \phi_{-}  + (d_n\tau + a_n\theta_{-})\phi_{-}^{\frac{n}{n-2}} + a_n\left( \frac{1}{2}\vert\theta_{-}\vert - r_{n}\tau\right)v^{\frac{2n}{n-2}}\phi_{-}^{-\frac{n}{n-2}}\geq 0,
\end{align*}
for all $f\in B_{M_f}$ and $X\in B_{M_{X}}$. Since we are assuming $\gamma$ to be Yamabe positive, we can deform $\gamma$ conformally to a metric with zero scalar curvature and zero boundary mean curvature. Thus, we can begin by assuming that $R_{\gamma}=0$ and $H_{\gamma}=0$ on $\partial M$. Now, since $\tau\in W^p_{1,\delta+1}$, we get that $\tau^2\in L^p_{\delta+2}$, also, since $(a_n\vert\theta_{-}\vert - d_n\tau)\geq 0$, we can define $\varphi_{-}\in W^p_{2,\delta}$ as the unique solution to
\begin{align}
\begin{split}
\Delta_{\gamma}\varphi_{-}-b_{n}\tau^2\varphi_{-}&=b_{n}\omega \tau^2,\\
-\hat{\nu}(\varphi_{-})-(a_n\vert\theta_{-}\vert - d_n\tau)\varphi_{-}&=\omega (a_n\vert\theta_{-}\vert - d_n\tau),
\end{split}
\end{align}
where $b_{n}\doteq \frac{n-1}{n}c_n$, and $\omega$ is a harmonic function with homogeneous Neumann boundary conditions, asymptotic to positive constants $\{A_j \}_{j=1}^N$ on each end. Then, define $\phi_{-}\doteq \alpha(\omega + \varphi_{-})$, where $\alpha>0$ is a constant to be fixed. Notice that 
\begin{align*}
\Delta_{\gamma}\phi_{-}&=\alpha\Delta_{\gamma}\varphi_{-}=\alpha b_n\tau^2(\omega+\varphi_{-})=b_{n}\tau^2\phi_{-},\\
-\hat{\nu}(\phi_{-})&=-\alpha\hat{\nu}(\varphi_{-})=\alpha (a_n\vert\theta_{-}\vert - d_n\tau)(\omega + \varphi_{-})= (a_n\vert\theta_{-}\vert - d_n\tau)\phi_{-},
\end{align*}
thus implying that
\begin{align}
\begin{split}
\Delta_{\gamma}\phi_{-}-b_n\tau^2\phi_{-}&=0,\\
-\hat{\nu}(\phi_{-})- (a_n\vert\theta_{-}\vert - d_n\tau)\phi_{-}&=0.
\end{split}
\end{align}
Then, because of the weak maximum principle given in \cite[Lemma A.1]{Holst2}, we known that $\phi_{-}\geq 0$, and then the strong maximum principle, given in \cite[Lemma 4]{Maxwell1}, guarantees that $\phi_{-}> 0$. Now, consider
\begin{align*}
\mathcal{H}^1_{f,X}(\phi_{-})&= b_n\tau^2\phi_{-} - b_n \tau^2\phi_{-}^{\frac{n+2}{n-2}} + 2c_n\epsilon_1\phi_{-}^{\frac{n+2}{n-2}} + c_n \vert\tilde{K}(X)\vert^2_{\gamma}\phi_{-}^{-\frac{3n-2}{n-2}} + 2c_n\epsilon_2(f)\phi_{-}^{-3} \\
&+ 2c_n\epsilon_3\phi_{-}^{\frac{n-6}{n-2}},\\
\mathcal{H}^2_{f,X}(\phi_{-})&=(a_n\vert\theta_{-}\vert - d_n\tau)\phi_{-} - (a_n\vert\theta_{-}\vert - d_n\tau)\phi_{-}^{\frac{n}{n-2}} + \left( \frac{1}{2}\vert\theta_{-}\vert - r_{n}\tau\right)v^{\frac{2n}{n-2}}\phi_{-}^{-\frac{n}{n-2}}
\end{align*}
Since $\sup_M(\omega+\varphi_{-})< \infty$, then, there exists $0<\alpha\ll 1$ such that
\begin{align*}
\alpha(\omega+\varphi_{-}) - \alpha^{\frac{n+2}{n-2}}(\omega+\varphi_{-})^{\frac{n+2}{n-2}}\geq 0 \text{ on } M,\\
\alpha(\omega+\varphi_{-}) - \alpha^{\frac{n}{n-2}}(\omega+\varphi_{-})^{\frac{n}{n-2}}\geq 0 \text{ on } \partial M,
\end{align*}
In fact, it is enough to take $\alpha$ sufficiently small so as to satisfy the following two conditions: 
\begin{align*}
\alpha^{\frac{n+2}{n-2}-1}=\alpha^{\frac{4}{n-2}}\leq \inf_M(\omega+\varphi_{-})^{1-\frac{n+2}{n-2}}=\inf_M(\omega+\varphi_{-})^{-\frac{4}{n-2}},\\
\alpha^{\frac{n}{n-2}-1}=\alpha^{\frac{2}{n-2}}\leq \inf_{\partial M}(\omega+\varphi_{-})^{1-\frac{n}{n-2}}=\inf_{\partial M}(\omega+\varphi_{-})^{-\frac{2}{n-2}}
\end{align*}
Notice that, since $0<\omega+\varphi_{-}\in C^0$ and it tends to some positive constants $\{A_j \}_{j=1}^{N}$ at infinity in each end $E_j$, then $0<\inf_M\;(\omega+\varphi_{-})^{-1}<\infty$. Such choice of $\alpha>0$ guarantees that $\mathcal{H}^{1,2}_{f,X}(\phi_{-})\geq 0$ $\forall$ $f,X\in W^p_{2,\delta}$, proving that $\phi_{-}$ is a strong global subsolution for the hamiltonian constraint.

Now, let us consider the supersolution. In this case, we need to find $\phi_{+}$ satisfying
\begin{align*}
\mathcal{H}_{f,X}(\phi_{+})\leq 0, \;\; \forall \;\; f\in B_{M_f}, X\in B_{M_{X}},
\end{align*}
for some radii $M_f,M_X>0$. In order to do this, let $\Lambda\in L^p_{\delta+2}$ be a positive function which agrees with $r^{-(\delta+\frac{n}{p})}r^{-(\delta+2+{\frac{n}{p}})}$ in a neighbourhood of infinity, in each end, and $\lambda\in W^p_{1-\frac{1}{p}}(\partial M)$ a positive function on the boundary. Define $\varphi_{+}\in W^p_{2,\delta}$ as the unique solution to 
\begin{align*}
\Delta_{\gamma}\varphi_{+}&=-\Lambda,\\
\hat{\nu}(\varphi_{+})&=\lambda \text{ on } \partial M,
\end{align*} 
Then, define $\phi_{+}\doteq \beta(\omega+\varphi_{+})$, where $\beta$ is some positive constant to be determined, and notice that 
\begin{align*}
\Delta_{\gamma}\phi_{+}&=-\beta\Lambda\leq 0,\\
-\hat{\nu}(\phi_{+})&=-\beta\lambda\leq 0 \text{ on } \partial M,
\end{align*}
which, appealing to the maximum principles given in \cite[Lemma A.1]{Holst2} and \cite[Lemma 4]{Maxwell1}, implies that $\phi_{+}> 0$. Being aware of the existence of strong global subsolutions of the form $\phi_{-}=\alpha(\omega+\varphi_{-})$ whenever $\alpha$ is sufficiently small, let us fix a relation between $\alpha$ and $\beta$ such that $\phi_{-}<\phi_{+}$. In order to do this, whatever $\beta$ is, consider $0<\alpha< \beta\inf_M\frac{\omega+\varphi_{+}}{\omega+\varphi_{-}}$. Now, consider the following
\begin{align*}
\mathcal{H}^{1}_{f,X}(\phi_{+})
&\leq -\beta\Lambda + 2c_n\epsilon_1\phi_{+}^{\frac{n+2}{n-2}} + c_n \vert\tilde{K}(X)\vert^2_{\gamma}\phi_{+}^{-\frac{3n-2}{n-2}} + 2c_n\epsilon_2(f)\phi_{+}^{-3} + 2c_n\epsilon_3\phi_{+}^{\frac{n-6}{n-2}},\\
\mathcal{H}^{2}_{f,X}(\phi_{+})&=- \beta\lambda - (a_n\vert\theta_{-}\vert - d_n\tau)\phi_{+}^{\frac{n}{n-2}} + \left( \frac{1}{2}\vert\theta_{-}\vert - c_{n}\tau\right)v^{\frac{2n}{n-2}}\phi_{+}^{-\frac{n}{n-2}},
\end{align*}
Since, given numbers $a,b$, it holds $(a+b)^2\leq 2a^2 + 2b^2$, then we get that
\begin{align*}
\vert\tilde{K}\vert^2_{\gamma}\leq 2 \vert\pounds_{\gamma,conf}X\vert^2_{\gamma} + 2\vert U\vert^2_{\gamma}.
\end{align*}
Now, fix the numbers
\begin{align}\label{radii}
\begin{split}
M_{f}&=C \Big\{\Vert\tilde{q}\Vert_{L^p_{\delta+2}}\Vert\phi_{+}^{\frac{2n}{n-2}}\Vert_{C^0} + \Vert E_{\hat{\nu}}\Vert_{W^p_{1-\frac{1}{p}}} \Big\},\\
M_{X}&=\kappa\Big\{\Vert D\tau\Vert_{L^p_{\delta+2}} \Vert\phi_{+}^{\frac{2n}{n-2}}\Vert_{C^0} +  \Vert\omega_1\Vert_{L^p_{\delta+2}} \Vert\phi_{+}^{2\frac{n + 1}{n-2}}\Vert_{C^0} + \Vert\tilde{F}\Vert_{W^p_{1,\delta+1}}\left(M_f  + \Vert\vartheta\Vert_{W^p_{1,\delta+1}}\right)\\
&  + \Vert \vert\theta_{-}\vert - 2r_{n}\tau\Vert_{W^p_{1-\frac{1}{p}}}\Vert v^{\frac{2n}{n-2}}\Vert_{C^0} + \Vert U\Vert_{W^p_{1-\frac{1}{p}}}\Big\},
\end{split}
\end{align}
where  $C$ and $\kappa$ are the constant appearing in the estimates (\ref{f-estimate}) and (\ref{Y-estimate}) respectively, and consider the balls $B_{M_f}$ and $B_{M_X}$ in $W^p_{2,\delta}$. Then, since $\vert\pounds_{\gamma,conf}X\vert_{\gamma}\lesssim \vert DX\vert_{\gamma}$, we get that for any $X\in B_{M_X}$, it holds that 
\begin{align*}
\vert\pounds&_{\gamma,conf}X\vert^2_{\gamma}\lesssim r^{-2(\delta+1+\frac{n}{p})}\Vert X\Vert^2_{W^p_{2,\delta}},\\
&\lesssim \kappa^2r^{-2(\delta+1+\frac{n}{p})}\Big\{ \Vert D\tau\Vert_{L^p_{\delta+2}} \Vert\phi_{+}^{\frac{2n}{n-2}}\Vert_{C^0} +  \Vert\omega_1\Vert_{L^p_{\delta+2}} \Vert\phi_{+}^{2\frac{n + 1}{n-2}}\Vert_{C^0} \\
& + \Vert\tilde{F}\Vert_{W^p_{1,\delta+1}}\left(M_f  + \Vert\vartheta\Vert_{W^p_{1,\delta+1}}\right) + \big\Vert \vert\theta_{-}\vert - 2r_{n}\tau\Vert_{W^p_{1-\frac{1}{p}}}\Vert v^{\frac{2n}{n-2}}\Vert_{C^0} + \Vert U\Vert_{W^p_{1-\frac{1}{p}}} \Big\}^2.
\end{align*}
We would like to separate the terms involving $D\tau$ from those not involving it. In order to do this, notice that the above implies that
\begin{align*}
\vert\pounds_{\gamma,conf}X\vert^2_{\gamma}&\lesssim \kappa^2r^{-2(\delta+1+\frac{n}{p})}\Big\{ \Vert D\tau\Vert^2_{L^p_{\delta+2}} \Vert\phi_{+}^{\frac{4n}{n-2}}\Vert_{C^0} +  \Big\{ \Vert \tilde{q}\Vert_{L^p_{\delta+2}}\Vert\tilde{F}\Vert_{W^{p}_{1,\delta+1}}\Vert\phi_{+}^{2\frac{n+2}{n-2}}\Vert_{C^0} \\
& + \Vert\tilde{F}\Vert_{W^p_{1,\delta+1}}\Vert\vartheta\Vert_{W^p_{1,\delta+1}}  + \Vert\omega_1\Vert_{L^p_{\delta+2}} \Vert\phi_{+}^{2\frac{n + 1}{n-2}}\Vert_{C^0} + \Vert\tilde{F}\Vert_{W^p_{1,\delta+1}}\Vert E_{\hat{\nu}}\Vert_{W^p_{1-\frac{1}{p}}}\\
& + \Vert \vert\theta_{-}\vert - 2r_{n}\tau\Vert_{W^p_{1-\frac{1}{p}}}\Vert v^{\frac{2n}{n-2}}\Vert_{C^0} + \Vert U\Vert_{W^p_{1-\frac{1}{p}}} \Big\}^2\Big\}.
\end{align*}
This estimate implies that
\begin{align*}
\begin{split}
\!\!\!\!\!\!\vert\tilde{K}(X)\vert^2_{\gamma}&\lesssim r^{-2(\delta+1+\frac{n}{p})}\Big\{\kappa^2\Big\{ \Vert D\tau\Vert^2_{L^p_{\delta+2}} \Vert\phi_{+}^{\frac{4n}{n-2}}\Vert_{C^0} +  \Big\{ \Vert\tilde{q}\Vert_{L^{p}_{\delta+2}}\Vert\tilde{F}\Vert_{W^{p}_{1,\delta+1}}\Vert\phi_{+}^{2\frac{n+2}{n-2}}\Vert_{C^0} \\
& + \Vert\tilde{F}\Vert_{W^p_{1,\delta+1}}\Vert\vartheta\Vert_{W^p_{1,\delta+1}}  + \Vert\omega_1\Vert_{L^p_{\delta+2}} \Vert\phi_{+}^{2\frac{n + 1}{n-2}}\Vert_{C^0} + \Vert\tilde{F}\Vert_{W^p_{1,\delta+1}}\Vert E_{\hat{\nu}}\Vert_{W^p_{1-\frac{1}{p}}}\\
& + \Vert \vert\theta_{-}\vert - 2r_{n}\tau\Vert_{W^p_{1-\frac{1}{p}}}\Vert v^{\frac{2n}{n-2}}\Vert_{C^0} + \Vert U\Vert_{W^p_{1-\frac{1}{p}}} \Big\}^2\Big\} + \Vert U\Vert^2_{W^p_{1,\delta+1}}\Big\},
\end{split}
\end{align*}
where, again, we have used that, if $U\in W^p_{1,\delta+1}$, then $\vert U\vert_e\lesssim r^{-(\delta+1+\frac{n}{p})}\Vert U\Vert_{W^p_{1,\delta+1}}$. Now, going back to the estimate of the hamiltonian constraint, we get that
\begin{align*}
\mathcal{H}^{1}(\phi_{+})&\leq -\beta\Lambda  +  C_nr^{-2(\delta+1+\frac{n}{p})}\Vert D\tau\Vert^2_{L^p_{\delta+2}} \Vert\phi_{+}^{\frac{4n}{n-2}}\Vert_{C^0} \phi_{+}^{-\frac{3n-2}{n-2}} \\
&+  C_nr^{-2(\delta+1+\frac{n}{p})}\Big\{ \Vert\tilde{q}\Vert_{L^{p}_{\delta+2}}\Vert\tilde{F}\Vert_{W^{p}_{1,\delta+1}}\Vert\phi_{+}^{2\frac{n+2}{n-2}}\Vert_{C^0} + \Vert\omega_1\Vert_{L^p_{\delta+2}} \Vert\phi_{+}^{2\frac{n + 1}{n-2}}\Vert_{C^0}\\
& + \Vert\tilde{F}\Vert_{W^p_{1,\delta+1}}\Vert\vartheta\Vert_{W^p_{1,\delta+1}} + \Vert\tilde{F}\Vert_{W^p_{1,\delta+1}}\Vert E_{\hat{\nu}}\Vert_{W^p_{1-\frac{1}{p}}} + \Vert \vert\theta_{-}\vert - 2r_{n}\tau\Vert_{W^p_{1-\frac{1}{p}}}\Vert v^{\frac{2n}{n-2}}\Vert_{C^0}\\
& + \Vert U\Vert_{W^p_{1-\frac{1}{p}}}  \Big\}^2\phi_{+}^{-\frac{3n-2}{n-2}} + C_nr^{-2(\delta+1+\frac{n}{p})}\Vert U\Vert^2_{W^p_{1,\delta+1}}\phi_{+}^{-\frac{3n-2}{n-2}} + 2c_n\epsilon_1\phi_{+}^{\frac{n+2}{n-2}} \\
& + 2c_n\epsilon_2(f)\phi_{+}^{-3} + 2c_n\epsilon_3\phi_{+}^{\frac{n-6}{n-2}}.
\end{align*}
Also, recall that, for any $f\in B_{M_f}$ it holds that
\begin{align*}
\begin{split}
2\epsilon_2(f)&=\vert df + \vartheta\vert^2_{\gamma}\leq r^{-2(\delta+1+\frac{n}{2})}\Vert df + \vartheta\Vert^2_{W^p_{1,\delta+1}},\\
&\leq 2r^{-2(\delta+1+\frac{n}{2})}\left(2C^2\Vert \tilde{q}\Vert^2_{L^p_{\delta+2}}\Vert{\phi_{+}}^{\frac{2n}{n-2}}\Vert^2_{C^0} + 2C^2\Vert E_{\hat{\nu}}\Vert^2_{W^p_{1-\frac{1}{p}}} + \Vert\vartheta\Vert^2_{W^p_{1,\delta+1}} \right).
\end{split}
\end{align*}
Thus, in general, we get that
\begin{align*}
\!\!\!\!\mathcal{H}&^1(\phi_{+})\leq -\beta\Lambda  + C_nr^{-2(\delta+1+\frac{n}{p})}\Vert D\tau\Vert^2_{L^p_{\delta+2}} \Vert\phi_{+}^{\frac{4n}{n-2}}\Vert_{C^0} \phi_{+}^{-\frac{3n-2}{n-2}} \\
&+  C_nr^{-2(\delta+1+\frac{n}{p})}\Big\{ \Vert\tilde{q}\Vert_{L^{p}_{{\delta+2}}}\Vert\tilde{F}\Vert_{W^{p}_{1,\delta+1}}\Vert\phi_{+}^{2\frac{n+2}{n-2}}\Vert_{C^0}  + \Vert\omega_1\Vert_{L^p_{\delta+2}} \Vert\phi_{+}^{2\frac{n + 1}{n-2}}\Vert_{C^0}\\
& + \Vert\tilde{F}\Vert_{W^p_{1,\delta+1}}\Vert\vartheta\Vert_{W^p_{1,\delta+1}}  + \Vert\tilde{F}\Vert_{W^p_{1,\delta+1}}\Vert E_{\hat{\nu}}\Vert_{W^p_{1-\frac{1}{p}}} + \Vert \vert\theta_{-}\vert - 2r_{n}\tau\Vert_{W^p_{1-\frac{1}{p}}}\Vert v^{\frac{2n}{n-2}}\Vert_{C^0}\\
&  + \Vert U\Vert_{W^p_{1-\frac{1}{p}}}  \Big\}^2\phi_{+}^{-\frac{3n-2}{n-2}} + C_nr^{-2(\delta+1+\frac{n}{p})}\Vert U\Vert^2_{W^p_{1,\delta+1}}\phi_{+}^{-\frac{3n-2}{n-2}} + 2c_n\epsilon_1\phi_{+}^{\frac{n+2}{n-2}} \\
& + 2c_nr^{-2(\delta+1+\frac{n}{2})}\left(2C^2\left(\Vert\tilde{q}\Vert^2_{L^p_{\delta+2}}\Vert{\phi_{+}}^{\frac{2n}{n-2}}\Vert_{C^0} + \Vert E_{\hat{\nu}}\Vert^2_{W^p_{1-\frac{1}{p}}}\right) + \Vert\vartheta\Vert^2_{W^p_{1,\delta+1}} \right)\phi_{+}^{-3} \\
&+ 2c_n\epsilon_3\phi_{+}^{\frac{n-6}{n-2}},
\end{align*}
In order to construct a far from CMC supersolution, notice that
\begin{align*}
\Vert\phi_{+}^{\frac{4n}{n-2}}\Vert_{C^0} \phi_{+}^{-\frac{3n-2}{n-2}}=\beta^{\frac{n+2}{n-2}}\Vert(\omega + \varphi_{+})^{\frac{4n}{n-2}}\Vert_{C^0} (\omega + \varphi_{+})^{-\frac{3n-2}{n-2}}.
\end{align*}
Near infinity we have that $\Lambda= r^{-2(\delta+1+\frac{n}{p})}$. Thus, if we pick $\beta$ sufficiently small, independently of how large $\Vert D\tau\Vert_{L^{p}_{\delta+2}}$ might be, we get that
\begin{align*}
-\beta\Lambda  +  C_nr^{-2(\delta+1+\frac{n}{p})}\Vert D\tau\Vert^2_{L^p_{\delta+2}}\Vert\phi_{+}^{\frac{4n}{n-2}}\Vert_{C^0} \phi_{+}^{-\frac{3n-2}{n-2}}< 0.
\end{align*} 
In fact, supposing $\Vert D\tau\Vert\neq 0$, we just need to satisfy
\begin{align*}
0<\beta^{\frac{n+2}{n-2}-1}=\beta^{\frac{4}{n-2}}<\inf_M\frac{\Lambda r^{2(\delta+1+\frac{n}{p})}}{C_n\Vert D\tau\Vert^2_{L^p_{\delta+2}}\Vert(\omega + \varphi_{+})^{\frac{4n}{n-2}}\Vert_{C^0} }(\omega + \varphi_{+})^{\frac{3n-2}{n-2}}
\end{align*}
Finally, notice that, if we consider $\mu\in W^p_{1,2\delta+2+\frac{n}{p}}$ being a non-negative function, then
\begin{align}\label{supersol2}
\begin{split}
\!\!\!\!\epsilon_1&=\mu\left( 1 + \vert\tilde{u}\vert^2_{\gamma} \right)\lesssim r^{-2(\delta+1+\frac{n}{p})}\Vert\mu\Vert_{W^p_{1,2\delta+2+\frac{n}{p}}}\left( 1 + \vert\tilde{u}\vert^2_{\gamma} \right),\\
\epsilon_3&=\frac{1}{4}\vert\tilde{F}\vert^2_{\gamma}\lesssim r^{-2(\delta+1+\frac{n}{p})}\Vert\tilde{F}\Vert^2_{W^p_{1,\delta+1}},\\
\Vert{\omega_1}\Vert_{L^p_{\delta+2}}&\leq \Vert\mu\Vert_{L^p_{\delta+2}}\Vert\left( 1 + \vert\tilde{u}\vert^2_{\gamma} \right)^{\frac{1}{2}}\tilde{u}\Vert_{C^0}.
\end{split}
\end{align}
Thus, we finally get that
\begin{align*}
\mathcal{H}&^{1}(\phi_{+})\leq -\beta\Lambda  + C_nr^{-2(\delta+1+\frac{n}{p})}\Vert D\tau\Vert^2_{L^p_{\delta+2}} \Vert \phi_{+}^{\frac{4n}{n-2}}\Vert_{C^0} \phi_{+}^{-\frac{3n-2}{n-2}} \\
&+  C_nr^{-2(\delta+1+\frac{n}{p})}\Big\{ \Vert\tilde{q}\Vert_{L^{p}_{\delta+2}}\Vert\tilde{F}\Vert_{W^{p}_{1,\delta+1}}\Vert\phi_{+}^{2\frac{n+2}{n-2}}\Vert_{C^0}  + \Vert\omega_1\Vert_{L^p_{\delta+2}} \Vert\phi_{+}^{2\frac{n + 1}{n-2}}\Vert_{C^0}\\
& + \Vert\tilde{F}\Vert_{W^p_{1,\delta+1}}\Vert\vartheta\Vert_{W^p_{1,\delta+1}}  + \Vert\tilde{F}\Vert_{W^p_{1,\delta+1}}\Vert E_{\hat{\nu}}\Vert_{W^p_{1-\frac{1}{p}}} + \Vert \vert\theta_{-}\vert - 2r_{n}\tau\Vert_{W^p_{1-\frac{1}{p}}}\Vert v^{\frac{2n}{n-2}}\Vert_{C^0} \\
& + \Vert U\Vert_{W^p_{1-\frac{1}{p}}}  \Big\}^2\phi_{+}^{-\frac{3n-2}{n-2}} + C_nr^{-2(\delta+1+\frac{n}{p})}\Vert U\Vert^2_{W^p_{1,\delta+1}}\phi_{+}^{-\frac{3n-2}{n-2}}  \\
& + C^{(1)}_nr^{-2(\delta+1+\frac{n}{p})}\Vert\mu\Vert_{W^p_{1,2\delta+2+\frac{n}{p}}}\left( 1 + \vert \tilde{u}\vert^2_{\gamma} \right)\phi_{+}^{\frac{n+2}{n-2}} + C^{(3)}_nr^{-2(\delta+1+\frac{n}{p})}\Vert\tilde{F}\Vert^2_{W^p_{1,\delta+1}}\phi_{+}^{\frac{n-6}{n-2}} \\
&+ 2c_nr^{-2(\delta+1+\frac{n}{2})}\left(2C^2\left(\Vert\tilde{q}\Vert^2_{L^p_{\delta+2}}\Vert{\phi_{+}}^{\frac{2n}{n-2}}\Vert_{C^0} + \Vert E_{\hat{\nu}}\Vert^2_{W^p_{1-\frac{1}{p}}}\right) + \Vert\vartheta\Vert^2_{W^p_{1,\delta+1}} \right)\phi_{+}^{-3}, 
\end{align*}
Notice that the first line in the above estimate is negative from our choice of $\beta$, and since $\Lambda=r^{-2(\delta+1+\frac{n}{p})}$ near infinity, then all the terms in the other lines decay at this rate or faster. Thus, under smallness assumptions on $\Vert\tilde{q}\Vert_{L^p_{\delta+2}},\Vert\tilde{F}\Vert_{W^p_{1,\delta+1}}, \Vert\mu \Vert_{W^p_{1,\delta+1}},\Vert v\Vert_{C^0},\Vert E_{\hat{\nu}}\Vert_{W^p_{1-\frac{1}{p}}}$, $\Vert\vartheta\Vert_{W^p_{1,\delta+1}}$ and $\Vert U\Vert_{W^p_{1,\delta+1}}$, we can guarantee that the right-hand side of the above inequality is negative. Furthermore, since $- \beta\lambda - (a_n\vert \theta_{-}\vert - d_n\tau)\phi_{+}^{\frac{n}{n-2}}<0$ on $\partial M$, under smallness assumptions on $v\in W^p_{1-\frac{1}{p}}$, we get that $\mathcal{H}^{2}_{f,X}(\phi_{+})\leq 0$ for any $f\in B_{M_f}$ and any $X\in B_{M_X}$. All this implies that, under the present assumptions and with the choices of balls $B_{M_f}$ and $B_{M_X}$ made above, $\phi_{-}$ and $\phi_{+}$ form a compatible pair of strong global supersolution for the Hamiltonian constraint. 

Let us now show that the map $\mathcal{F}_{a,b}$ is invariant on the balls $B_{M_f}$ and $B_{M_X}$ for all $\phi\in [\phi_{-},\phi_{+}]_{C^0}$. That is, we need to see that if $(\phi,f,X)=\mathcal{F}_{a,b}(\bar{\phi},\bar{f},\bar{X})$ with $(\bar{\phi},\bar{f},\bar{X})\in [\phi_{-},\phi_{+}]_{C^0}\times B_{M_f}\times B_{M_X}$, then $f\in B_{M_f}$ and $X\in B_{M_{X}}$. From the construction of $\mathcal{F}_{a,b}$, the elliptic estimates (\ref{f-estimate})-(\ref{Y-estimate}) and the definitions of $M_f$ and $M_X$, we straightforwardly get that
\begin{align*}
\Vert f\Vert_{W^p_{2,\delta}}\leq M_f, \;\; ; \;\;\Vert X\Vert_{W^p_{2,\delta}}\leq M_X,
\end{align*}
which proves the claim.

Finally, notice that the construction of $\phi_{+}$ does not depend on $v\in W^p_{1-\frac{1}{p}}(\partial M)$. Therefore, we can construct $\phi_{+}$ as above, then define $v\doteq \phi_{+}\vert_{\partial M}$ and continue as above with minor changes. More explicitly, in this case, the smallness condition for $v$ in $H^{1}_{f,X}(\phi_{+})\leq 0$ can be replaced by a smallness condition on $\big\Vert \vert\theta_{-}\vert - 2r_{n}\tau\big\Vert_{W^p_{1-\frac{1}{p}}}$, and concerning the boundary condition $\mathcal{H}^{2}_{f,X}(\phi_{+})\leq 0$ in the above proof, notice that the choice $v\doteq \phi_{+}\vert_{\partial M}$ implies that 
\begin{align*}
\begin{split}
\mathcal{H}^{2}_{f,X}(\phi_{+})&=- \beta\lambda - (a_n\vert\theta_{-}\vert - d_n\tau)\phi_{+}^{\frac{n}{n-2}} + a_n\left( \frac{1}{2}\vert\theta_{-}\vert - r_{n}\tau\right)v^{\frac{2n}{n-2}}\phi_{+}^{-\frac{n}{n-2}},\\
&=- \beta\lambda  - \frac{1}{2} a_n \vert\theta_{-}\vert \phi_{+}^{\frac{n}{n-2}}<0,
\end{split}
\end{align*}
which implies that this choice is admissible.
\end{proof}

Appealing to the above construction, we have the follow far from CMC existence result for (\ref{Conformal-EMSystem.1})-(\ref{boundcondsystems})-(\ref{ElectricPotentialEq}).

\begin{thm}\label{Yamabepostiveexistence}
Let $(M,\gamma)$ be a $W^p_{2,\delta}$-Yamabe positive AE manifold with compact boundary $\partial M$, with $p>n$, $n\geq 3$, and $\delta>-\frac{n}{p}$. Also, consider $\tau\in W^p_{1,\delta+1}(M)$, $U\in W^p_{1,\delta+1}(M,T^0_2M)$, $\tilde{F}\in W^p_{1,\delta+1}(M,\Lambda^2TM)$, $\tilde{u}\in W^{p}_{1,\delta}(M,TM)$, $\mu\in W^p_{1,2\delta+2+\frac{n}{p}}(M), \tilde{q}\in L^p_{\delta+2}(M),\vartheta\in W^p_{1,\delta+1}(M, T^{*}M),\theta_{-}\in W^p_{1-\frac{1}{p}}(\partial M)$ and $E_{\hat{\nu}}\in W^p_{1-\frac{1}{p}}(\partial M)$. If $\frac{1}{2}\vert\theta_{-}\vert- c_n\tau\geq 0$ and $\theta_{-}<0$ along $\partial M$ and $U,\tilde{F},\mu, \tilde{q},E_{\hat{\nu}},\vartheta$ and $\big\Vert \vert\theta_{-}\vert - 2r_{n}\tau\big\Vert_{W^p_{1-\frac{1}{p}}}$ are sufficiently small, then, there is a $W^p_{2,\delta}$-solution to the conformal problem (\ref{Conformal-EMSystem.1})-(\ref{boundcondsystems}). 
\end{thm}
\begin{proof}
First of all, from Lemma \ref{EMcompactness} we know that our system is conformal Einstein-type elliptic. Thus, let us begin by fixing a harmonic function $\omega$ tending to positive constants $\{A_j \}_{j=1}^N$ on each end $\{E_j\}_{j=1}^N$ respectively. Then, from Lemma \ref{Yamabeposbarriers} we know that under our smallness assumptions we can produce a compatible pair of strong global barriers $\phi_{\pm}$ asymptotic to $\alpha\omega$ and $\beta\omega$ for sufficiently small constants $\alpha\leq \beta$. Also, from the above lemma, we can fix the function $v\doteq \phi_{+}\vert_{\partial M}$ in (\ref{boundcondsystems}), and we know that the map $\mathcal{F}_{a,b}$ associated to the shifted system (\ref{shiftedsystem})-(\ref{shiftedboundarycond}) is invariant on the balls $B_{M_f}$ and $B_{M_X}$ in $W^p_{2,\delta}$. Then, from Theorem \ref{conformalexistence}, we know that we can get a solution to (\ref{Conformal-EMSystem.1})-(\ref{boundcondsystems})-(\ref{ElectricPotentialEq}) as long as $v\doteq \phi_{+}\vert_{\partial M}\geq \phi\vert_{\partial M}$. But from Theorem \ref{conformalexistence} this conditions follows trivially since the solution $\phi$ satisfies a uniform bound given by $\phi_{+}$.
\end{proof}

We will now produce strong global barriers for Lichnerowicz equation in the cases where $\gamma$ is not Yamabe positive. Again, we will follow Holst \textit{et al.} \cite{Holst2} closely in this construction. 

\begin{lemma}\label{Yamabenonposbarriers}
Let $(M,\gamma)$ be a $W^p_{2,\delta}$-AE manifold, with $p>n$, $n\geq 3$ and $-\frac{n}{p}<\delta<n-2-\frac{n}{p}$, suppose that $b_n\tau^2+c_nR_{\gamma}\geq 0$ and that $\frac{1}{2}\vert\theta_{-}\vert-r_n\tau\geq 0$, $v>0$ and $a_n(H+\vert\theta_{-}\vert) - d_n\tau\geq 0$ along $\partial M$. Assuming the same functional hypotheses as in Lemma \ref{Yamabeposbarriers} and under smallness assumptions on $\mu,\tilde{q},\tilde{F},\mu,\theta, E_{\hat{\nu}}$ and $\Vert d\tau\Vert_{L^p_{\delta+2}}$, there is a compatible pair of strong global barriers $0<\phi_{-}\leq\phi_{+}$ for the Lichnerowicz equation associated to charged dust, such that $\phi_{\pm}-\omega_{\pm}\in W^p_{2,\delta}$ for some harmonic functions $\omega_{\pm}$, satisfying Neumann boundary conditions, which tend to positive constants $\{A_j^{\pm}\}_{j=1}^N$ on each end. Also, the map $\mathcal{F}_{a,b}$ associated to the shifted system is invariant on balls $B_{M_f},B_{M_X}\subset W^p_{2,\delta}$ where the strong global barriers work for any $(a,b)\in L^p_{\delta+2}(M)\times W^{p}_{1-\frac{1}{p}}(\partial M)$ such that $a,b\geq 0$. Furthermore, if $\big\Vert \vert\theta_{-}\vert - 2r_{n}\tau\big\Vert_{W^p_{1-\frac{1}{p}}}$ is sufficiently small and replacing $a_n(H+\vert\theta_{-}\vert) - d_n\tau\geq 0$ by $H+\frac{\vert\theta_{-}\vert}{2}\geq 0$, the choice $v\doteq \phi_{+}\vert_{\partial M}$ is compatible with this construction.
\end{lemma} 
\begin{proof}
We will first produce the strong global subsolution. In order to do this, consider the equation
\begin{align}\label{eq1}
\begin{split}
\Delta_{\gamma}u &= c_n R_{\gamma}u + b_n\tau^2u^{\frac{n+2}{n-2}},\\
-\hat{\nu}(u) &= a_nHu + (a_n\vert\theta_{-}\vert - d_n\tau)u^{\frac{n}{n-2}} \text{ on } \partial M.
\end{split}
\end{align}
Following \cite[Theorem 5.7]{Holst2}, we can construct a positive solution $u\in W^p_{2,loc}$ to the above equation, which is asymptotic to some harmonic function $\omega_{-}$ which tends to positive constants $\{A^{-}_j\}_{j=1}^N$ on each end and satisfies Neumann boundary conditions. Then, consider $\phi_{-}\doteq \alpha u$ with $\alpha>0$ a constant to be fixed. Now, consider

\begin{align*}
\mathcal{H}^{1}_{f,X}(\phi_{-})
&= \left(\alpha - \alpha^{\frac{n+2}{n-2}}\right)b_n\tau^2u^{\frac{n+2}{n-2}} + 2c_n\epsilon_1\phi_{-}^{\frac{n+2}{n-2}} + c_n \vert\tilde{K}(X)\vert^2_{\gamma}\phi_{-}^{-\frac{3n-2}{n-2}}  \\
& + 2c_n\epsilon_2(f)\phi_{-}^{-3} + 2c_n\epsilon_3\phi_{-}^{\frac{n-6}{n-2}},
\end{align*} 
thus, if $\alpha<1$, we get that for any $f,X\in W^p_{2,\delta}$, it holds that $\mathcal{H}^{1}_{f,X}(\phi_{-})\geq 0$. Similarly, for such a choice of $\alpha<1$, we have that
\begin{align*}
\mathcal{H}^2_{f,X}(\phi_{-})&=(\alpha - \alpha^{\frac{n}{n-2}})(a_n\vert\theta_{-}\vert - d_n\tau)u^{\frac{n}{n-2}} + a_n\left( \frac{1}{2}\vert\theta_{-}\vert - r_{n}\tau\right)v^{\frac{2n}{n-2}}\phi_{-}^{-\frac{n}{n-2}}\geq 0,
\end{align*}
for any $f\in B_{M_f}$ and $X\in B_{M_X}$. Finally, since $\phi_{-}>0$ is bounded on $M$, let us fix $\alpha>0$ being sufficiently small so that $\phi_{-}\leq 1$.  

Now, in order to present a global supersolution we will proceed as follows. Let $r$ be a smooth positive function on $M$, which, in the ends, near infinity, agrees with the euclidean radial function $\vert x\vert$. Then, define $B=\bar{B}r^{-(\delta+\frac{n}{p})}r^{-(\delta+2+\frac{n}{p})}\in L^p_{\delta+2}$, where $\bar{B}$ is a positive constant to be fixed later, and consider the equation
\begin{align}\label{B-eq}
\begin{split}
\Delta_{\gamma}\phi_B&= - B\phi_B^{-\frac{3n-2}{n-2}},\\
-\hat{\nu}(\phi_B)&=- a_nB \text{ on } \partial M. 
\end{split}
\end{align}
and notice that $v_{-}=1$ is a subsolution. Define the function $u\in W^p_{2,\delta}$ as the unique solution to
\begin{align*}
\Delta_{\gamma} u&= - B,\\
-\hat{\nu}(u)&=- a_nB \text{ on } \partial M.
\end{align*}
Then, define $v_{+}\doteq\beta(1+u)$ and notice that $\Delta_{\gamma}v_{+}\leq 0$, $\hat{\nu}(v_{+})\vert_{\partial M}=-a_n\beta B\vert_{\partial M}\leq 0$ and $v_{+}$ tends to $\beta>0$ at infinity, thus, because of the maximum principles, $v_{+}>0$. Furthermore, since $u\in C^0$ is bounded, we can pick $\beta$ sufficiently large so that $v_{+}>1$. In particular, consider $\beta>1$ sufficiently large such that $v_{+}>v_{-}$. Then, we get that 
\begin{align*}
\Delta_{\gamma}v_{+}&= - \beta B \leq   -\beta Bv_{+}^{-\frac{3n-2}{n-2}}\leq -Bv_{+}^{-\frac{3n-2}{n-2}},\\
-\hat{\nu}(v_{+})&= -\beta a_nB\leq - a_n B,
\end{align*}  
showing that $v_{+}$ is a supersolution of (\ref{B-eq}). Then, from \cite[Theorem A.4]{Holst2}, we know that for any harmonic function $\omega_{+}$ satisfying Newman boundary conditions, and tending to constants $\{A_j\}_{j=1}^N$ on each end satisfying $1\leq A_j\leq\beta$, there is a positive solution to the equation (\ref{B-eq}) which is asymptotic to $\omega_{+}$ and trapped between $v_{-}$ and $v_{+}$. Let us call such a solution by $\phi_{+}=\omega_{+}+\varphi_{+}\geq 1\geq \phi_{-}$ with $\varphi_{+}\in W^p_{2,\delta}$ and show that we can build strong global barriers using $\phi_{+}$. Notice that
\begin{align*}
\mathcal{H}^{1}_{f,X}(\phi_{+})
&\leq - c_n R_{\gamma} -  b_n\tau^2 + \left(c_n \vert\tilde{K}(X)\vert^2_{\gamma} - B\right)\phi_{+}^{-\frac{3n-2}{n-2}}  + 2c_n\epsilon_2(f) \phi_{+}^{-3} \\
&+ 2c_n\epsilon_1\phi_{+}^{\frac{n+2}{n-2}} + 2c_n\epsilon_3\phi_{+}^{\frac{n-6}{n-2}}.
\end{align*}
Notice that from our hypotheses $- c_n R_{\gamma} -  b_n\tau^2\leq 0$. Furthermore, we know that there are constants $C_1,C_2>0$ such that
\begin{align*}
c_n \vert\tilde{K}(X)\vert^2_{\gamma}&\leq C_1 r^{-2(\delta+1+\frac{n}{p})}\left(\Vert X\Vert^2_{W^{p}_{2,\delta}} + \Vert U\Vert^2_{W^{p}_{1,\delta+1}} \right),\\
2c_n\epsilon_2(f)&\leq C_2 r^{-2(\delta+1+\frac{n}{p})}\left(\Vert f\Vert^2_{W^{p}_{2,\delta}} + \Vert\vartheta\Vert^2_{W^{p}_{1,\delta+1}} \right).
\end{align*}
Thus, if we consider balls $B_{M_X}$ and $B_{M_f}$ in $W^p_{2,\delta}$ of radii $M_{X}$ and $M_f$ respectively, we get that for any $X\in B_{M_{X}}$ and $f\in B_{M_{f}}$ it holds that
\begin{align*}
c_n \vert\tilde{K}(X)\vert^2_{\gamma}&\leq C_1 r^{-2(\delta+1+\frac{n}{p})}\left(M^2_{X} + \Vert U\Vert^2_{W^{p}_{1,\delta+1}} \right),\\
2c_n\epsilon_2(f)&\leq C_2 r^{-2(\delta+1+\frac{n}{p})}\left(M^2_f + \Vert\vartheta\Vert^2_{W^{p}_{1,\delta+1}} \right).
\end{align*}
Let us first fix these radii by assuming that the have the same form as in (\ref{radii}), but now with our present functions $\phi_{-}$ and $\phi_{+}$. Then, the constant $\bar{B}$ is to be chosen satisfying
\begin{align*}
\bar{B}>&C_1\left(M^2_{X} + \Vert U\Vert^2_{W^{p}_{1,\delta+1}} \right).
\end{align*}
This choice would guarantee that $B>c_n \vert\tilde{K}(X)\vert^2_{\gamma}$ for any $X\in B_{M_X}$ and $f\in B_{M_f}$ with $B\in L^p_{\delta+2}$. Thus, let us write the above condition explicitly:
\begin{align}\label{B-bound}
\begin{split}
\!\!\!\!\bar{B}>&C_1\Big\{\kappa^2\Big\{\Vert D\tau\Vert_{L^p_{\delta+2}} \Vert\phi_{+}^{\frac{2n}{n-2}}\Vert_{C^0} +  \Vert\omega_1\Vert_{L^p_{\delta+2}} \Vert\phi_{+}^{2\frac{n + 1}{n-2}}\Vert_{C^0} +  C\Vert\tilde{q}\Vert_{L^p_{\delta+2}}\Vert\tilde{F}\Vert_{W^p_{1,\delta+1}}\Vert\phi_{+}^{\frac{2n}{n-2}}\Vert_{C^0}\\
& + C\Vert E_{\hat{\nu}}\Vert_{W^p_{1-\frac{1}{p}}}\Vert\tilde{F}\Vert_{W^p_{1,\delta+1}} + \Vert\tilde{F}\Vert_{W^p_{1,\delta+1}}\Vert\vartheta\Vert_{W^p_{1,\delta+1}} + \Vert \vert\theta_{-}\vert - 2r_{n}\tau\Vert_{W^p_{1-\frac{1}{p}}}\Vert v^{\frac{2n}{n-2}}\Vert_{C^0} \\
& + \Vert U\Vert_{W^p_{1-\frac{1}{p}}}\Big\}^2 + \Vert U\Vert^2_{W^{p}_{1,\delta+1}}\Big\}.
\end{split}
\end{align}
Notice that since $\phi_{+}$ depends on $B$, we cannot choose $\bar{B}$ without affecting the right hand side of the above expression. Nevertheless, we can make the choice of $\bar{B}$ satisfying
\begin{align}\label{B-bound2}
\begin{split}
\bar{B}>&C_1\Big\{\kappa^2\Big\{C\Vert E_{\hat{\nu}}\Vert_{W^p_{1-\frac{1}{p}}}\Vert\tilde{F}\Vert_{W^p_{1,\delta+1}} + \Vert\tilde{F}\Vert_{W^p_{1,\delta+1}}\Vert\vartheta\Vert_{W^p_{1,\delta+1}}  \\
&+ \Vert \vert\theta_{-}\vert - 2r_{n}\tau\Vert_{W^p_{1-\frac{1}{p}}}\Vert v^{\frac{2n}{n-2}}\Vert_{C^0} + \Vert U\Vert_{W^p_{1-\frac{1}{p}}}\Big\}^2 + \Vert U\Vert^2_{W^{p}_{1,\delta+1}}\Big\}\doteq \bar{B}_1,
\end{split}
\end{align}
which only depends on known data and then, using this choice, if $\Vert D\tau\Vert_{L^p_{\delta+2}}, \Vert\omega_1\Vert_{L^p_{\delta+2}}$ and $\Vert\tilde{q}\Vert_{L^p_{\delta+2}}$
are small enough, we can guarantee that (\ref{B-bound}) holds. Notice that the smallness condition on $\omega_1$ is actually translated into a smallness condition on $\mu$. Taking into account these considerations, we get that for any $X\in B_{M_X}$ and $f\in B_{M_f}$ the following holds.
\begin{align*}
\!\!\!\!\mathcal{H}&^{1}_{f,X}(\phi_{+})\leq - c_n R_{\gamma} -  b_n\tau^2 + r^{-2(\delta+1+\frac{n}{p})}\left(C_1 \left(M^2_{X} + \Vert U\Vert^2_{W^{p}_{1,\delta+1}} \right) - \bar{B}_1\right)\phi_{+}^{-\frac{3n-2}{n-2}}\\
&+ C_1r^{-2(\delta+1+\frac{n}{p})}\left(M^2_f + \Vert\vartheta\Vert^2_{W^{p}_{1,\delta+1}} \right)\phi_{+}^{-3} \\
&+ C_n^{(2)}r^{-2(\delta+1+\frac{n}{p})}\Vert \mu\Vert_{W^p_{1,2\delta+2+\frac{n}{p}}}\left( 1 + \vert\tilde{u}\vert^2_{\gamma} \right)\phi_{+}^{\frac{n+2}{n-2}} + C_n^{(3)}r^{-2(\delta+1+\frac{n}{p})}\Vert\tilde{F}\Vert^2_{W^p_{1,\delta+1}}\phi_{+}^{\frac{n-6}{n-2}}.
\end{align*}
The first line in the above relation is strictly negative from our choice of $\bar{B}$, thus, under smallness assumptions on $\mu,\tilde{q},\tilde{F}, E_{\hat{\nu}}$ and $\vartheta$, we get that $\mathcal{H}^{1}_{f,X}(\phi_{+})\leq 0$ for any $f\in B_{M_f}$ and $X\in B_{M_X}$. Similarly, since $\phi_{+}\geq 1$, under our hypotheses it holds that
\begin{align*}
\mathcal{H}^{2}_{f,X}(\phi_{+})&=-a_nB - a_n H \phi_{+}  - (a_n\vert\theta_{-}\vert-d_n\tau)\phi_{+}^{\frac{n}{n-2}} + a_n\left( \frac{1}{2}\vert\theta_{-}\vert - r_{n}\tau\right)v^{\frac{2n}{n-2}}\phi_{+}^{-\frac{n}{n-2}},\\
&\leq -a_nB - (a_n (H  + \vert\theta_{-}\vert)- d_n\tau) + a_n\left( \frac{1}{2}\vert\theta_{-}\vert - r_{n}\tau\right)v^{\frac{2n}{n-2}}.
\end{align*}
Since $-(a_n (H  + \vert\theta_{-}\vert)- d_n\tau)\leq 0$ on $\partial M$, then, if $\bar{B}$ satisfies
\begin{align*}
\bar{B}>\sup_{\partial M}\left(\left( \frac{1}{2}\vert\theta_{-}\vert - r_{n}\tau\right)v^{\frac{2n}{n-2}}r^{2(\delta+1+\frac{n}{p})}\right)\doteq \bar{B}_2,
\end{align*}
we get that $\mathcal{H}^2_{f,X}(\phi_{+})\leq 0$ for any $f\in B_{M_f}$ and any $X\in B_{M_X}$. Therefore, choosing $\bar{B}\doteq \max\{\bar{B}_1,\bar{B}_2 \}$, we can guarantee that $\phi_{+}$ is a strong global supersolution. 

Now, similarly to what we did in Lemma \ref{Yamabeposbarriers}, the elliptic estimates associated to the momentum and electromagnetic constraints, guarantee that the map 
\begin{align*}
[\phi_{-},\phi_{+}]_{C^0}\times B_{M_f}\times B_{M_X} &\mapsto W^p_{2,loc}\times B_{M_f}\times B_{M_X},\\
(\bar{\phi},\bar{f},\bar{X})&\xrightarrow[]{\mathcal{F}_{a,b}} (\phi,f,X),
\end{align*}
is invariant on $B_{M_f}\times B_{M_X}$ for any $\bar{\phi}\in [\phi_{-},\phi_{+}]_{C^0}$.

Finally, in order to show that the choice $v=\phi_{+}\vert_{\partial M}$ is admissible, we need two minor modifications in the above proof. First, when choosing $\bar{B}$ in the relations (\ref{B-bound})-(\ref{B-bound2}) now $v=\phi_{+}\vert_{\partial M}$ also depends on $\bar{B}$, therefore, we must consider 
\begin{align*}
\begin{split}
\bar{B}>&C_1\Big\{\kappa\Big\{C\Vert E_{\hat{\nu}}\Vert_{W^p_{1-\frac{1}{p}}}\Vert\tilde{F}\Vert_{W^p_{1,\delta+1}} + \Vert\tilde{F}\Vert_{W^p_{1,\delta+1}}\Vert\vartheta\Vert_{W^p_{1,\delta+1}} + \Vert U\Vert_{W^p_{1-\frac{1}{p}}}\Big\} + \Vert U\Vert_{W^{p}_{1,\delta+1}}\Big\}\doteq \bar{B}_1,
\end{split}
\end{align*}
and demand that $\Vert \vert\theta_{-}\vert - 2r_{n}\tau\Vert_{W^p_{1-\frac{1}{p}}}$ is small enough so that we can keep $\mathcal{H}^1_{f,X}(\phi_{+})\leq 0$ for all $f\in B_{M_f}$ and $X\in B_{M_X}$. Secondly, regarding the boundary condition and similarly to final step in the proof of Lemma \ref{Yamabeposbarriers}, if $v=\phi_{+}\vert_{\partial M}$ and $\phi_{+}\geq 1$, then
\begin{align*}
\mathcal{H}^2_{f,X}(\phi_{+})\leq-a_nB-a_n(H+\frac{\vert \theta_{-}\vert}{2})\phi_{+}\leq 0
\end{align*}
as long as $H+\frac{\vert\theta_{-}\vert}{2}\geq 0$.
\end{proof}

The proof of the following theorem now runs in complete analogy to the proof of Theorem \ref{Yamabepostiveexistence}. In this case we get existence for AE-manifold outside the Yamabe positive case, but with a constraint on the possible choices of $\gamma$ and $\tau$: We can accommodate arbitrary chosen $W^p_{2,\delta}$-AE metrics at the expense of creating regions of large mean curvature. 

\begin{thm}\label{Yamabenonpostiveexistence}
Let $(M,\gamma)$ be a $W^p_{2,\delta}$-AE manifold with compact boundary $\partial M$, with $p>n$, $n\geq 3$, and $\delta>-\frac{n}{p}$. Also, consider $\tau\in W^p_{1,\delta+1}(M)$, $U\in W^p_{1,\delta+1}(M,T^0_2M)$, $\tilde{F}\in W^p_{1,\delta+1}(M,\Lambda^2TM)$, $\tilde{u}\in W^{p}_{1,\delta}(M,TM)$, $\mu\in W^p_{1,2\delta+2+\frac{n}{p}}(M), \tilde{q}\in L^p_{\delta+2}(M),\vartheta\in W^p_{1,\delta+1}(M, T^{*}M)$, $\theta_{-}\in W^p_{1-\frac{1}{p}}(\partial M)$ and $E_{\hat{\nu}}\in W^p_{1-\frac{1}{p}}(\partial M)$. If $\mu, \tilde{q}, \tilde{F}, E_{\hat{\nu}},\vartheta$ and $\big\Vert\vert\theta_{-}\vert-2c_n\tau\big\Vert_{W^p_{1-\frac{1}{p}}}$ are sufficiently small; $c_n R_{\gamma} + b_n\tau^2\geq 0$; $\theta_{-}\leq 0$ and $\frac{1}{2}\vert\theta_{-}\vert-c_n\tau\geq 0$ and $H+\frac{\vert\theta_{-}\vert}{2}\geq 0$ along $\partial M$, then there is a $W^p_{2,\delta}$-solution to the conformal problem (\ref{Conformal-EMSystem.1})-(\ref{boundcondsystems}).
\end{thm}

\end{document}